\documentclass[reqno, 10pt]{amsart}

\usepackage[top=4cm, bottom=3cm, left=3cm, right=3cm]{geometry}
\usepackage{amsmath}

\newtheorem{thm}{Theorem}
\newtheorem{remark}{Remark}
\newtheorem{cor}{Corollary}
\newtheorem{lem}{Lemma}
\newtheorem{prop}{Proposition}
\theoremstyle{definition}

\theoremstyle{remark}

\newcommand{\U}{\mathcal{U}}
\newcommand{\N}{\mathcal{N}}

\newcommand{\F}{\mathcal{F}}

\newcommand{\cA}{\mathcal{A}}

\newcommand{\cL}{\mathcal{L}}
\newcommand{\cK}{\mathcal{K}}
\newcommand{\cH}{\mathcal{H}}
\newcommand{\cU}{\mathcal{U}}
\newcommand{\cF}{\mathcal{F}}
\newcommand{\cN}{\mathcal{N}}

\newcommand{\bx}{{\bf x}}
\newcommand{\bR}{{\mathbb R}}
\newcommand{\bC}{{\mathbb C}}
\newcommand{\bN}{{\mathbb N}}

\newcommand{\tr}{\mbox{Tr}}
\newcommand{\ph}{\varphi}

\newcommand{\wt}{\widetilde}

\numberwithin{equation}{section} \numberwithin{lem}{section}
\numberwithin{thm}{section} \numberwithin{prop}{section}
\numberwithin{cor}{section} \numberwithin{rem}{section}

\begin{document}

\thispagestyle{headings}

\title{Rate of Convergence Towards Hartree Dynamics}

\author{Li Chen}
\address{Department of Mathematical Sciences, Tsinghua University, Beijing, 100084, People's Republic of China} \email{{\tt lchen@math.tsinghua.edu.cn}}
\thanks{Li Chen is partially supported by National Natural Science Foundation of China (NSFC), grant number 10871112}

\author{Ji Oon Lee}
\address{Department of Mathematical Sciences, Korea Advanced Institute of Science and Technology, Daejeon, 305701, Republic of Korea}
\email{jioon.lee@kaist.edu}

\author{Benjamin Schlein}
\address{Institute for Applied Mathematics, University of Bonn, Endenicher Allee 60, 53115 Bonn, Germay}
\email{benjamin.schlein@hcm.uni-bonn.de}
\thanks{Benjamin Schlein is partially supported by an ERC Starting Grant}

\maketitle

\begin{abstract}
We consider a system of $N$ bosons interacting through a two-body potential with, possibly, Coulomb-type singularities. We show that the difference between the many-body Schr\"odinger evolution in the mean-field regime and the effective nonlinear Hartree dynamics is at most of the order $1/N$, for any fixed time. The $N$-dependence of the bound is optimal.  
\end{abstract}

\section{Introduction}

A system of $N$ bosons in three dimensions is described by $\psi_N \in L^2_s (\bR^{3N} , dx_1 \dots dx_N)$, the subspace of $L^2 (\bR^{3N}, dx_1 \dots dx_N)$ consisting of functions which are invariant with respect to permutations of the $N$ particles (the invariance w.r.t. permutation expresses the bosonic symmetry; fermionic systems are described by antisymmetric wave functions). We always assume $\psi_N$ to be normalized so that $\| \psi_N \|_2 = 1$ ($|\psi_N (x_1, \dots , x_N)|^2$ is interpreted as the probability density for finding particles close to $(x_1, \dots , x_N)$). We consider Hamilton operators with two-body interactions, having the form 
\begin{equation}\label{eq:ham0} H_{N,\lambda} = \sum_{j=1}^N -\Delta_{x_j} + \lambda \sum_{i<j}^N V (x_i -x_j) \end{equation}
and acting as self-adjoint operators on the Hilbert space $L^2_s (\bR^{3N} , dx_1 \dots dx_N)$. In (\ref{eq:ham0}), the sum of the Laplacians is the kinetic energy of the $N$ particles, $\lambda \in \bR$ is a coupling constant, and the sum of $V(x_i -x_j)$ over all pairs of particles describes the potential energy of the system ($V(x_i -x_j)$ acts as a multiplication operator); appropriate conditions on $V$ will be specified below. 

The evolution of the system is governed by the $N$ particle Schr\"odinger equation 
\begin{equation}\label{eq:schr} i \partial_t \psi_{N,t} = H_{N,\lambda} \psi_{N,t} \,. \end{equation}
The solution of the Schr\"odinger equation can be obtained by applying the unitary group generated by $H_{N,\lambda}$ to the initial wave function $\psi_{N, t=0}$; in other words, (\ref{eq:schr}) is always solved by $\psi_{N,t} = e^{-i H_{N,\lambda} t} \psi_{N, 0}$. In this sense, establishing existence and uniqueness of solutions of (\ref{eq:schr}) is not an issue. What makes the study of (\ref{eq:schr}) challenging is the fact that, in systems of interest in physics, the number of particles $N$ involved in the evolution is typically huge ($N$ ranges from values of the order $10^3$ in extremely dilute samples of Bose-Einstein condensates, up to values of the order $10^{23}$ in chemical samples). For such values of $N$, the expression $\psi_{N,t} = e^{-i H_{N,\lambda} t} \psi_{N,0}$ is not useful if one is interested in establishing quantitative or even qualitative properties of the dynamics. For this reason, one of the main goals of quantum statistical mechanics is the derivation of effective evolution equations which, on the one hand, can be approached by numerical methods (in contrast with (\ref{eq:schr})), and, on the other hand, approximate the solution of (\ref{eq:schr}) in the interesting regimes. 

One of the simplest regime where effective evolution equations can be used to approximate the full many-body evolution is the so-called mean field limit, which is characterized by a large number of very weak collisions among the particles. To realize the mean field limit, we consider large values of $N$ (many collisions) and small values of the coupling constant $\lambda$ (weak interactions). A non-trivial effective dynamics can only emerge when the many collisions produce a total force of order one on each particle; in other words, when $N \lambda$ is of order one. To study the mean-field regime, we set therefore $\lambda = 1/N$ and we consider the evolution generated by the Hamiltonian 
\begin{equation}\label{eq:ham-mf} H_N = \sum_{j=1}^N -\Delta_{x_j} + \frac{1}{N} \sum_{i<j}^N V (x_i - x_j) \end{equation}
in the limit of large $N$. In particular, we are interested in the evolution of factorized initial wave functions of the form $\psi_N = \ph^{\otimes N}$ (here, we use the notation 
$\ph^{\otimes N}  (x_1, \dots , x_N) = \prod_{j=1}^N \ph (x_j)$). Because of the interaction, factorization is not preserved by the time-evolution. However, since collisions are very weak, we may still expect that factorization is approximately preserved in the limit of large $N$. In other words, we may expect that, for large $N$, the solution $\psi_{N,t} = e^{-i H_N t} \psi_N$ of the Schr\"odinger equation can be approximated (in a sense to be made precise later), by 
\begin{equation}\label{eq:fact} \psi_{N,t}  \simeq \ph_t^{\otimes N} \end{equation}
for a suitable one-particle wave function $\ph_t$. Assuming (\ref{eq:fact}) to be correct, it is simple to derive a self-consistent equation for the one-particle orbital $\ph_t$. In fact, (\ref{eq:fact}) implies that, at time $t \in \bR$, particles are distributed in space, independently of each other, with probability density $|\ph_t|^2$. This means that the potential experienced by a particle at $x \in \bR^3$ can be approximated by the average, mean field, potential $(V* |\ph_t|^2) (x)$ and therefore, that $\ph_t$ must satisfy the nonlinear Hartree equation
\begin{equation}\label{eq:hartree0} i \partial_t \ph_t = -\Delta \ph_t + (V * |\ph_t|^2) \ph_t \end{equation}
with initial data $\ph_{t=0} = \ph$. 

In which sense can we expect the solution of the $N$-particle Schr\"odinger equation $\psi_{N,t}$ to be approximated by the factorized wave function on the r.h.s. of (\ref{eq:fact})? It turns out that one cannot expect convergence in norm (see, however, the recent works \cite{GMM,GMM2} where second order corrections to the mean-field dynamics are taken into account to obtain a norm approximation of the full dynamics). Instead, (\ref{eq:fact}) has to be understood on the level of the reduced density matrices. Let $|\psi_{N,t} \rangle \langle \psi_{N,t}|$ denote the orthogonal projection onto $\psi_{N,t}$. Then, for $k=1,\dots, N$ we define the $k$-particle reduced density matrix by taking the partial trace of $|\psi_{N,t} \rangle \langle \psi_{N,t}|$ over the degrees of freedom associated with the last $N-k$ particles, that is \[ \gamma^{(k)}_{N,t} = \tr_{k+1, \dots , N} \,  |\psi_{N,t} \rangle \langle \psi_{N,t} | \, .\] 
In other words, $\gamma^{(k)}_{N,t}$ is defined as the non-negative trace class operator on $L^2 (\bR^{3k}, dx_1, \dots dx_k)$ with the kernel 
\begin{equation}\label{eq:marg} \begin{split} \gamma^{(k)}_{N,t} ( \bx_k ; \bx'_k) &= \int dx_{k+1} \dots  dx_N \, \psi_{N,t}  (\bx_k , x_{k+1}, \dots, x_N) \overline{\psi}_{N,t} (\bx'_k, x_{k+1}, \dots, x_N) \end{split} \end{equation} 
where we set $\bx_k = (x_1, \dots , x_k)$ and, similarly, $\bx'_k = (x'_1, \dots , x'_k)$. 
{F}rom the normalization $\| \psi_{N,t} \| = 1$, we conclude that $\tr \, \gamma^{(k)}_{N,t} = 1$ for all $1 \leq k \leq N$ and all $t \in \bR$. Observe that, for $1 \leq k < N$, the $k$-particle reduced density $\gamma^{(k)}_{N,t}$ does not contain the full information about the $N$-particle system.  Nevertheless, knowledge of $\gamma^{(k)}_{N,t}$ is sufficient to compute the expectation of $k$-particle observables, that is of observables of the form $O^{(k)} \otimes 1^{(N-k)}$ which only act non-trivially on $k$ particles. In fact, 
\[ \left\langle \psi_{N,t} , \left( O^{(k)} \otimes 1^{(N-k)} \right) \psi_{N,t} \right\rangle = \tr\, \gamma_{N,t} \, \left( O^{(k)} \otimes 1^{(N-k)} \right) = \tr \, \gamma^{(k)}_{N,t} O^{(k)} \, . \]
It turns out that the reduced density matrices are the right quantities to 
understand (\ref{eq:fact}). For a large class of potentials $V$, one can show that the reduced density matrices associated with $\psi_{N,t}$ converge, in the limit of large $N$, to the reduced density matrices associated with the factorized wave function $\ph_t^{\otimes N}$. In other words, one can show that, for any fixed $t \in \bR$, 
\begin{equation}\label{eq:diff} \tr \left| \gamma^{(1)}_{N,t} - |\ph_t\rangle \langle \ph_t| \right| \to 0 \end{equation} as $N \to \infty$. Observe that convergence of the one-particle density towards a rank-one projection immediately implies convergence of higher order reduced densities as well; for any fixed $k \in \bN$ and $t\in \bR$, it follows from (\ref{eq:diff}) that $\gamma^{(k)}_{N,t} \to |\ph_t \rangle \langle \ph_t|^{\otimes k}$ as $N \to \infty$ in the trace norm topology (see Remark~\ref{rem} below). 

The convergence (\ref{eq:diff}) has first been established by Spohn in \cite{Sp} for bounded potentials. In \cite{EY}, Erd\"os and Yau extended the techniques of Spohn to prove (\ref{eq:diff}) for potentials with a Coulomb-type singularity $V(x) = \pm 1/ |x|$ (partial results in this direction were also obtained in \cite{BGM}). In \cite{RS}, (\ref{eq:diff}) was established again for potentials with Coulomb singularities. In contrast with the previous results, the bound obtained in \cite{RS} gives an explicit estimate on the rate of the convergence. For factorized initial data, it is shown in \cite{RS} that
\begin{equation}\label{eq:rate} \tr\; \left| \gamma^{(1)}_{N,t} - |\ph_t \rangle \langle \ph_t| \right| \leq \frac{C e^{Kt}}{\sqrt{N}} \end{equation}
for constants $C,K$ depending only on the initial one-particle wave function $\ph$. The approach used in \cite{RS} is based on techniques first introduced by Hepp in \cite{He} and then extended by Ginibre and Velo in \cite{GV} for the study of the related problem of the classical limit of quantum mechanics. More recently, bounds of the form (\ref{eq:rate}) on the rate of convergence of the Schr\"odinger evolution towards the Hartree dynamics, were obtained by Knowles and Pickl in \cite{KP}, for potential with singularities of the form $|x|^{-\alpha}$, for $\alpha < 5/2$ (the bound on the rate of convergence obtained in \cite {KP} deteriorates, compared to (\ref{eq:rate}), for potentials of the form $|x|^{-\alpha}$, with $\alpha > 1$). In \cite{ES,MS}, the convergence (\ref{eq:diff}) was established for particles with a relativistic dispersion (the kinetic energy $-\Delta_{x_j}$ is replaced by $\sqrt{1-\Delta_{x_j}}$, for $j=1,\dots , N$) and with Coulomb type interaction $V(x) = \pm \lambda / |x|$ (this situation is physically interesting because it describes systems of gravitating bosons, so called boson stars, and the related phenomenon of stellar collapse). In order to describe the dynamics of Bose-Einstein condensates, it is interesting to consider, in (\ref{eq:ham-mf}), two-body potentials which scale with the number of particles $N$, and tend to a delta-function in the limit of large $N$. In this regime, the many-body quantum dynamics is approximated by the Gross-Pitaevskii equation; this problem has been studied in \cite{ESY1,ESY2,ESY3,P}. 

In this paper, we extend the techniques developed in \cite{RS}, and we improve the bound (\ref{eq:rate}) on the rate of convergence towards the Hartree dynamics. For interaction potentials with Coulomb type singularities and for factorized initial wave functions, we show that the difference between the reduced one-particle density associated with the solution of the $N$-particle Schr\"odinger equation and the orthogonal projection onto the solution of the Hartree equation (\ref{eq:hartree0}) is at most of the order $1/N$, for any fixed time $t \in \bR$. The $N$-dependence of this bound is expected to be optimal. Note that the same bound on the rate of convergence was obtained in \cite{ErS}, for bounded potentials, and, more recently, in \cite{CL} under the condition that $V \in L^3 (\bR^3) + L^{\infty} (\bR^3)$, which excludes a Coulomb type singularity. The main result of this paper is the following theorem. 

\begin{thm} \label{thm:main}
Suppose that the potential $V (x)$ satisfies the operator inequality
\begin{equation} \label{eq:Vcond}
V (x)^2 \leq D (1-\Delta_x)
\end{equation}
for some constant $D>0$. Let $\ph \in H^1 (\bR^3)$ with $\| \ph \|_2 = 1$, and let $\ph_t$ be the solution of the Hartree equation
\begin{equation}\label{eq:hartree} i\partial_t \ph_t = -\Delta \ph_t + (V*|\ph_t|^2) \ph_t \end{equation} with initial data $\ph_{t=0} = \ph$.  Let $\psi_{N,t} = e^{-iH_N t} \ph^{\otimes N}$ and $\gamma_{N, t}^{(1)}$ be the one-particle reduced density associated with $\psi_{N,t}$, as defined in \eqref{eq:marg}. Then, there exist constants $C$ and $K$, depending only on $\| \varphi \|_{H^1}$ and $D$, such that
\begin{equation} \label{eq:main}
\textrm{Tr } \Big| \gamma_{N, t}^{(1)} - | \varphi_t \rangle \langle \varphi_t | \Big| \leq \frac{C e^{Kt}}{N}.
\end{equation}
\end{thm}

\begin{remark}\label{rem}
The same techniques used to show (\ref{eq:main}) can be extended to prove an analogous bound for higher order reduced densities; for any fixed $k \in \bN$, one can show the existence of constants $C_k ,K_k$ such that
\[ \tr \; \left| \gamma^{(k)}_{N,t} - |\ph_t \rangle \langle \ph_t|^{\otimes k} \right| \leq \frac{C_k e^{K_k t}}{N} \, . \]
Note that, if one is satisfied with a slower rate of convergence for higher order reduced densities, a simple argument, outlined in Section 2 of \cite{KP}, shows that, for any $k \in \bN$, 
\[ \tr \, \big|\gamma_{N,t}^{(k)}-|\ph_t \rangle\langle\ph_t|^{\otimes k} \,\big| \leq 
2 \sqrt{2k \tr \, \big| \gamma_{N,t}^{(1)} - |\ph_t \rangle \langle \ph_t| \big|}  \leq \frac{C k^{1/2} e^{K |t|}}{\sqrt{N}} \, \] for constants $C,K$ independent of $k$. 
\end{remark}

\begin{remark}
With exactly the same techniques used to prove Theorem \ref{thm:main}, one can also consider mean-field Hamiltonians with external potential $V_{\text{ext}} (x)$ acting on the $N$ particles. The requirements on $V_{\text{ext}}$ are minimal (the conditions to make sure that the dynamics exists).
\end{remark}

As in \cite{RS,CL}, the main challenge to prove Theorem \ref{thm:main} consists in controlling the fluctuations around the mean-field dynamics. After second quantization, the evolution of these fluctuations is described by a two parameter group of unitary transformations $\cU (t;s)$ (see (\ref{eq:U})). It turns out that the growth of the fluctuations can be bounded by comparing first $\cU (t;s)$ with a simpler approximate dynamics $\cU_2 (t;s)$ having a quadratic generator $\cL_2 (t)$ (see (\ref{eq:U2}) and (\ref{eq:L2})). A similar approach was already used in \cite{RS}; to get the optimal bound on the fluctuations, however, we need to consider here, similarly to \cite{CL}, a different approximate dynamics. The problem reduces then to estimating the difference between the two evolutions $\cU (t;s)$ and $\cU_2 (t;s)$. While in \cite{CL} this difference was bounded using Strichartz-type estimates (requiring $V \in L^3 (\bR^3) + L^{\infty} (\bR^3)$ and therefore excluding Coulomb-type singularity), in the present paper we make use of an a-priori bound on the growth of the kinetic energy with respect to the approximate dynamics $\cU_2 (t;s)$. It turns out that, for this a-priori bound to be useful, we have to introduce a small, $N$-dependent, cutoff $\alpha_N$ in the interaction $V$; for sufficiently small $\alpha_N$, we show that the error due to the cutoff decays faster than $1/N$ and can therefore be absorbed in the right hand side of~(\ref{eq:main}). 

The paper is organized as follows. First, in Section \ref{sec:reg}, we show that the many body Schr\"odinger evolution with cutoffed potential remains close to the Schr\"odinger evolution with full potential $V$ (and, similarly, that the Hartree dynamics with regularized interaction remains close to the full Hartree dynamics), if the cutoff tends to zero sufficiently fast as $N \to \infty$. Hence, Theorem \ref{thm:main} follows by proving the corresponding bound for the difference between the regularized Schr\"odinger evolution and the regularized Hartree dynamics; this crucial bound is stated in Proposition~\ref{prop:main}. In Section \ref{sec:fock}, we define the bosonic Fock space and we recall some of its properties. In Section \ref{sec:proof}, we reformulate the convergence problem on the Fock space, and we prove Proposition~\ref{prop:main} making use of a series of estimates (in particular, the a-priori bound for the growth of the kinetic energy, which follows by combining Lemma \ref{lm:7} and Lemma \ref{lm:L2}) deferred to Sections \ref{sec:bdN}-\ref{sec:U2}.

\section{Regularization of the interaction}
\label{sec:reg}

For an arbitrary sequence $\alpha_N > 0$, we set  
\begin{equation}\label{eq:wtV}
\widetilde V (x) = \text{sgn}(V (x)) \cdot \min \{ |V (x)|, \alpha_N^{-1} \} 
\end{equation}
where $\text{sgn} (V(x))$ denotes the sign of $V (x)$. We also define the regularized Hamiltonian 
\begin{equation}\label{eq:regham}
\widetilde{H}_N = \sum_{j=1}^N -\Delta_{x_j} + \frac{1}{N} \sum_{i<j}^N \widetilde{V}(x_i -x_j). \end{equation}
Note that, by definition $|\wt{V} (x)| \leq \alpha_N^{-1}$. Moreover, (\ref{eq:Vcond})  implies the operator inequality
\begin{equation}\label{eq:wtVcond} \wt{V}^2 (x) \leq D (1-\Delta_x) \, .\end{equation}

Instead of proving directly Theorem \ref{thm:main}, we show that it is enough to prove the corresponding statement for the dynamics generated by the regularized Hamiltonian (\ref{eq:regham}), if $\alpha_N$ converges to zero sufficiently fast. First, we bound the difference between the evolution of the initial $N$-particle wave function $\psi_N$ w.r.t.  $H_N$ and w.r.t. the regularized Hamiltonian $\wt{H}_N$. 
\begin{lem}\label{lm:compare1}
Let $\psi_N = \ph^{\otimes N}$ for some $\ph \in H^1 (\bR^3)$ with $\| \ph \| = 1$. Let $\psi_{N,t} = e^{-iH_N t} \psi_N$ and $\wt{\psi}_{N,t} = e^{-i\widetilde{H}_N t} \psi_N$. Then there exists a universal constant $C>0$ such that 
\begin{equation}
\left\| \psi_{N,t} - \wt{\psi}_{N,t} \right\|^2 \leq C N \alpha_N \, |t| \,
\end{equation}
for all $N \in \bN$, $t \in \bR$.
\end{lem}

\begin{proof}
We consider the derivative
\begin{equation}\label{eq:ddtpsi} \begin{split}
\frac{d}{dt} \left\| \psi_{N,t} - \wt{\psi}_{N,t} \right\|^2 &= -2 \text{Re } \frac{d}{dt} \, \langle \psi_{N,t} , \wt{\psi}_{N,t} \rangle \\
&= 2 \, \text{Im } \langle (H_N - \widetilde{H}_N) \psi_{N,t} , \wt{\psi}_{N,t} \rangle \\
&= \frac{2}{N} \sum_{i<j}^N \text{Im } \left\langle \left(V(x_i - x_j) - \widetilde{V}(x_i -x_j) \right) \psi_{N,t}, \wt{\psi}_{N,t} \right\rangle \, .
\end{split} \end{equation}
Observe that the definition (\ref{eq:wtV}) of $\widetilde{V}$ implies that 
\begin{equation}\label{eq:V-V}
|V-\widetilde{V}| \leq |V| \cdot \mathbf{1} (|V| \geq \alpha^{-1}_N) \leq |V|^2 \alpha_N.
\end{equation}
Hence, from (\ref{eq:ddtpsi}), we obtain (using also the assumption (\ref{eq:Vcond})) 
\begin{equation} \label{eq:ddtpsi2} \begin{split}
\left| \frac{d}{dt} \left\| \psi_{N,t} - \wt{\psi}_{N,t} \right\|^2 \right|  \leq &C N \left|\left\langle  \left(V(x_1 - x_2) - \widetilde{V}(x_1 -x_2) \right) \psi_{N,t} , \wt{\psi}_{N,t} \right\rangle \right|\\
\leq & C N \alpha_N \left\langle \psi_{N,t} , (1-\Delta_{x_1}) \psi_{N,t} \right\rangle^{1/2} \left\langle \wt{\psi}_{N,t} , (1-\Delta_{x_1}) \wt{\psi}_{N,t} \right\rangle^{1/2} \,.
\end{split}
\end{equation}
Next, we note that, again from (\ref{eq:Vcond}),  
\[ \begin{split} 
N \langle \psi_{N,t} , (1- \Delta_{x_1}) \psi_{N,t} \rangle &\leq C \langle \psi_{N,t} , (H_N + N) \psi_{N,t} \rangle  \\ &\leq C \langle \ph^{\otimes N} , (H_N + N) \ph^{\otimes N} \rangle \\ &\leq C N \| \ph \|^2_{H^1} \,. \end{split} \]
Similarly, from (\ref{eq:wtVcond}), 
\[ N \langle \wt{\psi}_{N,t} , (1- \Delta_{x_1}) \wt{\psi}_{N,t} \rangle \leq N \| \ph \|^2_{H^1} \,.  \]
Therefore (\ref{eq:ddtpsi2}) implies that 
\[ \left| \frac{d}{dt} \left\| \psi_{N,t} - \wt{\psi}_{N,t} \right\|^2 \right| \leq CN  \alpha_N \, . \]
The lemma follows after integrating over $t$.
\end{proof}

As a consequence, we obtain a bound on the difference between the marginal densities associated with $\psi_{N,t}$ and $\wt{\psi}_{N,t}$. 
\begin{cor}\label{prop:gm-wtgm}
For any $k \in \bN$, let $\gamma^{(k)}_{N,t}$ and $\wt{\gamma}^{(k)}_{N,t}$ be the $k$-particle reduced densities associated with $\psi_{N,t} = e^{-iH_N t} \ph^{\otimes N}$ and $\wt{\psi}_{N,t} = e^{-i \wt{H}_N t} \ph^{\otimes N}$. Suppose $\alpha_N \leq N^{-3}$ in the definition (\ref{eq:wtV}). Then there exists a constant $C>0$, independent of $k$,  such that 
\[ \tr \, \left| \gamma^{(k)}_{N,t} - \wt{\gamma}^{(k)}_{N,t} \right| \leq  \frac{C \, |t|^{1/2}}{N} \, . \] 
\end{cor}

\begin{proof}
We have
\begin{equation}
\tr \left| \gamma^{(k)}_{N,t} - \wt{\gamma}^{(k)}_{N,t} \right| = \sup_{\| O^{(k)} \| \leq 1} \left| \tr \, O^{(k)} \left( \gamma^{(k)}_{N,t} - \wt{\gamma}^{(k)}_{N,t} \right) \right|
\end{equation}
where the supremum is taken over all compact operators $O^{(k)}$ over $L^2 (\bR^{3k},  dx_1 \dots dx_k)$, with operator norm $\| O^{(k)} \| \leq 1$. Observe that
\begin{equation} \begin{split}
\tr \, O^{(k)} \left( \gamma^{(k)}_{N,t} - \wt{\gamma}^{(k)}_{N,t} \right) &= \langle \psi_{N,t}, (O^{(k)} \otimes 1) \, \psi_{N,t} \rangle - \langle \wt{\psi}_{N,t} , (O^{(k)} \otimes 1) \, \wt{\psi}_{N,t} \rangle \\
&= \langle (\psi_{N,t} - \wt{\psi}_{N,t}), (O^{(k)} \otimes 1) \psi_{N,t} \rangle + \langle \wt{\psi}_{N,t} , (O^{(k)} \otimes 1) \, (\psi_{N,t} - \wt{\psi}_{N,t} ) \rangle \, . 
\end{split} \end{equation}
Taking absolute value, we find
\begin{equation}
\left| \tr \, O^{(k)} \left( \gamma^{(k)}_{N,t} - \wt{\gamma}^{(k)}_{N,t} \right) \right| \leq 2 \| \psi_{N,t} - \wt{\psi}_{N,t} \|
\end{equation}
for all observables $O^{(k)}$ with $\| O^{(k)} \| \leq 1$.
The corollary follows then from Lemma \ref{lm:compare1}.
\end{proof}

Finally, we estimate the distance between the solutions of the nonlinear Hartree equations with the full potential $V$ and with the regularized potential $\wt{V}$. 
\begin{lem} \label{lm:comp-h}
Let $\ph \in H^1 (\bR^3)$. Let $\ph_t$ be the solution of the Hartree equation (\ref{eq:hartree}) and $\wt{\ph}_t$ the solution of the Hartree equation
\begin{equation}\label{eq:hartree-reg} i\partial_t \wt{\ph}_t = -\Delta \wt{\ph}_t + (\wt{V} * |\wt{\ph}_t|^2) \wt{\ph}_t \end{equation}
with regularized potential $\wt{V}$, with $\ph_{t=0} = \wt{\ph}_{t=0} = \ph$. Then
\begin{equation}\label{eq:ph-wtph} \| \ph_t - \wt{\ph}_t \| \leq C \alpha_N e^{K |t|} \,. \end{equation}
Therefore \begin{equation}\label{eq:gm-wtgm} \tr \; \left| |\ph_t \rangle \langle \ph_t |^{\otimes k} - |\wt{\ph}_t \rangle \langle \wt{\ph}_t |^{\otimes k} \right| \leq 2 k \, \| \ph_t - \wt{\ph}_t \| \leq C k \alpha_N e^{K |t|} \end{equation} for any $k \in \bN$.
\end{lem} 

\begin{proof}
{F}rom (\ref{eq:Vcond}) and (\ref{eq:wtVcond}) it is easy to check that $\| \ph_t \|_{H^1}, \| \wt{\ph}_t \|_{H^1} \leq C$, for a constant $C$ which only depends on $\| \ph \|_{H^1}$.  We compute
\begin{equation} \begin{split}
\frac{d}{dt} \left\| \ph_t - \wt{\ph}_t \right\|^2 = \; &2 \, \text{Im } \left\langle \ph_t, \left[ V* |\ph_t|^2 - \widetilde{V}* |\wt{\ph}_t|^2 \right] \wt{\ph}_t \right\rangle \\
= \; &2 \, \text{Im } \left\langle \ph_t, \left[\left( V - \widetilde{V}\right)* |\ph_t|^2 \right]  \wt{\ph}_t \right\rangle \\
&+2  \, \text{Im }
 \left\langle \ph_t,  \left[ \widetilde{V}* \left( |\ph_t|^2 - |\wt{\ph}_t|^2\right) \right] (\wt{\ph}_t - \ph_t) \right\rangle
\end{split}\end{equation}
where, in the last line, we used the fact that
\begin{equation}
\text{Im } \left\langle \ph_t, \left( \widetilde{V}* \left( |\ph_t|^2 - |\wt{\ph}_t|^2\right) \right) \ph_t  \right\rangle = 0 \,.
\end{equation}

Using (\ref{eq:V-V}) we find, taking the absolute value,
\begin{equation}\label{eq:gr1} \begin{split}
\left| \frac{d}{dt} \left\| \ph_t - \wt{\ph}_t \right\|^2 \right| \leq  \; & 2  \alpha_N \| \ph_t \| \| \wt{\ph}_t \| \, \sup_x \int dy V^2(x-y)\, |\ph_t (y)|^2 \\ & + 2  \, \| \ph_t - \wt{\ph}_t \| \, \| \ph_t\| \, 
\sup_x \int dy \, |\wt V(x-y)| \, | \ph_t (y) - \wt\ph_t (y)| ( |\ph_t (y)| + |\wt{\ph}_t (y)|) \\
\leq  \; & C \alpha_N  + C \, \| \ph_t - \wt{\ph}_t \|^2 \, . 
\end{split}\end{equation}
In the last inequality we used that, from (\ref{eq:wtVcond}), 
\begin{equation} \begin{split}
\int dy \, |\wt V(x-y)| \, & | \ph_t (y) - \wt\ph_t (y)| \, ( |\ph_t (y)| + |\wt{\ph}_t (y)|) \\
& \leq \, \left( \int dy \,  | \ph_t (y) - \wt\ph_t (y)|^2 \right)^{1/2} \left( \int dy \wt V^2(x-y) (|\ph_t (y)| + |\wt{\ph}_t (y)|)^2 \right)^{1/2}  \\
& \leq C \| \ph_t - \wt\ph_t \| \, (\| \ph_t \|_{H^1} + \| \wt\ph_t \|_{H^1}) \, . 
\end{split}\end{equation}
From \eqref{eq:gr1} we obtain (by Gronwall)
\begin{equation}
\left\| \ph_t - \wt{\ph}_t \right\|^2 \leq C \alpha_N (e^{C |t|} - 1)
\end{equation}
which concludes the proof of \eqref{eq:ph-wtph}. To show \eqref{eq:gm-wtgm}, we write
\begin{equation}
|\ph_t \rangle \langle \ph_t|^{\otimes k} - |\wt{\ph}_t \rangle \langle \wt{\ph}_t|^{\otimes k} = \sum_{j=1}^k |\ph_t \rangle \langle \ph_t|^{\otimes (j-1)} \otimes \left( |\ph_t \rangle \langle \ph_t | - |\wt{\ph}_t \rangle \langle \wt{\ph}_t| \right) \otimes |\wt{\ph}_t \rangle \langle \wt{\ph}_t|^{\otimes (k -j)}
\end{equation}
and we use the fact that
\begin{equation}
\tr \, \Big| |\ph_t \rangle \langle \ph_t | - |\wt{\ph}_t \rangle \langle \wt{\ph}_t| \Big| \leq
2 \| \ph_t - \wt\ph_t \| \,.
\end{equation}
\end{proof}
As a consequence of Corollary \ref{prop:gm-wtgm} and of Lemma \ref{lm:comp-h}, Theorem \ref{thm:main} follows from the next proposition, which only involves regularized dynamics (regularized $N$-particle evolution and regularized Hartree dynamics).
\begin{prop} \label{prop:main}
Let $\wt{V}$ be as in (\ref{eq:wtV}) with some $\alpha_N \geq N^{-r}$ for some $r \in \bN$. Let $\ph \in H^1 (\bR^3)$, $\wt \gamma_{N, t}^{(1)}$ the one-particle reduced density associated with $e^{-i \wt{H}_N t} \ph^{\otimes N}$ and $\wt{\ph}_t$ the solution of the regularized Hartree equation \begin{equation}\label{eq:hartreereg} i\partial_t \wt{\ph}_t = -\Delta \wt{\ph}_t + (\wt{V} * |\wt{\ph}_t|^2) \wt{\ph}_t   \end{equation} with initial data $\ph_{t=0} = \ph$. Then, there exist constants $C$ and $K$ such that
\begin{equation} \label{trace norm bound}
\textrm{Tr } \Big| \wt \gamma_{N, t}^{(1)} - | \wt\varphi_t \rangle \langle \wt \varphi_t | \Big| \leq \frac{C e^{Kt}}{N}.
\end{equation}
\end{prop}

The proof of this proposition is given in Section \ref{sec:proof}. It makes use of a representation of the problem on the bosonic Fock space, which we introduce 
in the next section.  

\section{Fock Space Representation}
\label{sec:fock}

The bosonic Fock space over $L^2 (\bR^3, dx)$ is defined as the Hilbert space 
\[ \cF = \bigoplus_{n \geq 0} L^2 (\bR^3, dx)^{\otimes_s n} = \bC \oplus \bigoplus_{n \geq 1} L^2_s (\bR^{3n} , dx_1 \dots dx_n) \, . \] 
Here $L^2_s (\bR^{3n}, dx_1 \dots dx_n)$ denotes the subspace of $L^2 (\bR^{3n}, dx_1 \dots dx_n)$ consisting of functions symmetric with respect to any permutation of the $n$ variables $x_1, \dots , x_n$. In other words, $\cF$ contains sequences $\psi = \{ \psi^{(n)} \}_{n \geq 0}$
of $n$-particle wave functions $\psi^{(n)} \in L^2_s (\bR^{3n}, dx_1 \dots dx_n)$.
For $\psi_1, \psi_2 \in \cF$, we define the scalar product  
\[ \langle \psi_1 , \psi_2 \rangle = \sum_{n \geq 0} \langle
\psi_1^{(n)} , \psi_2^{(n)} \rangle_{L^2 (\bR^{3n})} =
\overline{\psi_1^{(0)}} \psi_2^{(0)} + \sum_{n \geq 1} \int d x_1
\dots d x_n \, \overline{\psi_1^{(n)}} (x_1 , \dots , x_n)
\psi_2^{(n)} (x_1, \dots ,x_n) \,. \] A sequence $\{ 0, \dots ,0,  \psi^{(m)}, 0, \dots\}$ describes a state with exactly $m$ particles. We will denote by $\cF^{(m)}$ the $m$-particle sector of $\cF$, which is spanned by vectors of the form $\{ 0, \dots ,0,  \psi^{(m)}, 0, \dots\}$. The vector $\Omega = \{1, 0, 0, \dots \} \in \cF$ is known as the vacuum and spans the zero-particle sector $\cF^{(0)}$. 

The number of particles operator $\cN$ is defined on the Fock space $\cF$ by $(\cN
\psi)^{(n)} = n \psi^{(n)}$. For $f \in L^2 (\bR^3, dx)$ we define the
creation operator $a^* (f)$ and the annihilation operator $a(f)$ by
\begin{equation}
\begin{split}
\left(a^* (f) \psi \right)^{(n)} (x_1 , \dots ,x_n) &=
\frac{1}{\sqrt n} \sum_{j=1}^n f(x_j) \psi^{(n-1)} ( x_1, \dots,
x_{j-1}, x_{j+1},
\dots , x_n) \\
\left(a (f) \psi \right)^{(n)} (x_1 , \dots ,x_n) &= \sqrt{n+1} \int
d x \; \overline{f (x)} \, \psi^{(n+1)} (x, x_1, \dots ,x_n) \, .
\end{split}
\end{equation}
For any $f \in L^2 (\bR^3, dx)$, the operators $a^* (f)$ and $a(f)$ are unbounded, densely defined, closed operators. The creation operator $a^*(f)$ is the adjoint of
the annihilation operator $a(f)$ (note that by definition $a(f)$ is
anti-linear in $f$), and they satisfy the canonical commutation relations; for any $f,g \in L^2 (\bR^3,dx)$, 
\begin{equation}\label{eq:comm} [ a(f) , a^* (g) ] =
\langle f,g \rangle_{L^2 (\bR^3)}, \qquad [ a(f) , a(g)] = [ a^*
(f), a^* (g) ] = 0 \,. \end{equation} For every $f\in L^2 (\bR^3, dx)$,
we introduce the self adjoint operator
\begin{equation}\label{eq:phi} \phi (f) = a^* (f) + a(f) \,. \end{equation}
It is interesting to note that the $n$-particle product state $\{ 0, \dots, 0, f^{\otimes n}, 0, \dots \}$ can be produced starting from the vacuum $\Omega$ by applying the creation operator $a^* (f)$ for $n$ times. More precisely, we have
\begin{equation} \label{eq:factor} \{ 0, \dots , 0, f^{\otimes n}, 0, \dots \} = \frac{a^* (f)^n}{\sqrt{n!}} \Omega \, . \end{equation}
The normalization can be easily checked using the canonical commutation relation. We will be interested in the time evolution of these product states. 

We will also make use of operator valued distributions $a^*_x$ and
$a_x$ ($x \in \bR^3$), defined so that \begin{equation}\begin{split}
a^* (f) &= \int d x \, f(x) \, a_x^* \\ a(f) & = \int d x \,
\overline{f (x)} \, a_x \end{split}
\end{equation}
for every $f \in L^2 (\bR^3 , dx)$. The canonical commutation relations
take the form \[ [ a_x , a^*_y ] = \delta (x-y) \qquad [ a_x, a_y ] = [ a^*_x , a^*_y] = 0 \, .\]

For an operator $J$ acting on the one-particle space $L^2 (\bR^3, dx)$, we define the second quantization $d\Gamma (J)$ of $J$ as the operator on $\cF$ whose action on the $n$-particle sector is given by
\[ \left(d\Gamma (J) \psi \right)^{(n)} = \sum_{j=1}^n J_j \psi^{(n)} \] 
where $J_j = 1 \otimes \dots 1 \otimes J \otimes 1 \dots \otimes 1$ is the operator $J$ acting only on the $j$-th variable. As an example, the number operator is the second quantization of the identity, i.e. $\cN = d\Gamma (1)$. If the one-particle operator $J$ has a kernel $J(x;y)$, then the second quantization $d\Gamma (J)$ can be written in terms of the operator valued distributions $a_x, a_x^*$ as
\[ d\Gamma (J)  = \int dx dy \, J(x;y) \, a_x^* a_y \,. \]
For example, we have 
\[ \cN = \int dx \, a_x^* a_x  \,.\]

The following lemma provides useful bounds to control creation and annihilation operators as well as operators of the form $d\Gamma (J)$ 
in terms of the number of particle operator $\cN$.
\begin{lem}\label{lm:0} 
For $\alpha >0$, let $D(\cN^{\alpha}) = \{ \psi \in \cF : \sum_{n \geq 1} n^{2\alpha} \| \psi^{(n)} \|^2 < \infty \}$ denote the domain of the operator $\cN^{\alpha}$. For any $f \in L^2 (\bR^3, dx)$ and any $\psi \in D (\cN^{1/2})$, we have 
\begin{equation}\label{eq:bd-a} 
\begin{split}  \| a(f) \psi \| & \leq \| f \| \, \| \cN^{1/2} \psi \|, \\ 
\| a^* (f) \psi \| &\leq \| f \| \, \| (\cN+1)^{1/2} \psi \|, \\ 
\| \phi (f) \psi \| &\leq 2 \| f \| \| \left( \cN + 1 \right)^{1/2}
\psi \| \, . \end{split} \end{equation}
Moreover, for any bounded one-particle operator $J$ on $L^2 (\bR^3, dx)$ and for every $\psi \in D (\cN)$, we find 
\begin{equation}\label{eq:J-bd} \| d\Gamma (J) \psi \| \leq \| J \| \| \cN \psi \|  \, .\end{equation}
\end{lem}

\begin{proof}
These bounds are standard. A proof of (\ref{eq:bd-a}) can be found, for example, in Lemma 2.1 of \cite{RS}. As for (\ref{eq:J-bd}), it is enough to observe that 
\[ \| d\Gamma (J) \psi \|^2 = \sum_{n\geq 1} \sum_{i,j = 1}^n \langle J_i \psi^{(n)}, J_j \psi^{(n)} \rangle \leq \sum_{n \geq 1} n^2 \| J \|^2 \| \psi^{(n)} \|^2 = \| J \|^2 \| \cN \psi \|^2 \] because, clearly, $\| J_i \| = \|J \|$ for all $i = 1, \dots ,n$.
\end{proof}

Given $\psi \in \cF$, we define the one-particle density
$\gamma^{(1)}_{\psi}$ associated with $\psi$ as the positive trace
class operator on $L^2 (\bR^3, dx)$ with kernel given by
\begin{equation}\label{eq:margi} \gamma^{(1)}_{\psi} (x; y) = \frac{1}{\langle \psi,
\cN \psi \rangle} \, \langle \psi, a_y^* a_x \psi \rangle\, .
\end{equation} By definition, $\gamma_{\psi}^{(1)}$ is a positive trace
class operator on $L^2 (\bR^3, dx)$ with $\tr \, \gamma_{\psi}^{(1)}
=1$. For an arbitrary $N$-particle state $\{ 0, \dots , 0 , \psi_N , 0, \dots \}$ it is simple to check that (\ref{eq:margi}) coincides with the definition (\ref{eq:marg}) given in the introduction. 

\medskip

On $\cF$, we define the Hamilton operator $\cH_N$ by $ (\cH_N \psi)^{(n)} =
\cH^{(n)}_N \psi^{(n)}$, with
\[ \cH^{(n)}_N = - \sum_{j=1}^n \Delta_j + \frac{1}{N} \sum_{i<j}^n
\wt V(x_i -x_j) \, , \] where $\wt{V}$ denotes the regularized potential introduces in (\ref{eq:wtV}). Using the distributions $a_x, a^*_x$, $\cH_N$
can be rewritten as
\begin{equation}\label{eq:ham2} \cH_N = \int d x \nabla_x a^*_x
\nabla_x a_x + \frac{1}{2N} \int d x d y \, \wt V(x-y) a_x^* a_y^*
a_y a_x \, . \end{equation} 
It is clear that, on the $N$-particle sector, the operator $\cH_N$
coincides with the Hamiltonian $\wt H_N$ defined in (\ref{eq:regham}). 
To study the dynamics generated by the Hamiltonian $\cH_N$ on the Fock space, coherent states will be useful. 

\medskip

For $f \in L^2 (\bR^3, dx)$, we define the Weyl-operator
\begin{equation}
W(f) = \exp \left( a^* (f) - a(f) \right) = \exp \left( \int d x
\, (f(x) a^*_x - \overline{f} (x) a_x) \right) \, .
\end{equation}
The coherent state with one-particle wave function $f$ is the vector $W(f) \Omega$. Notice that \begin{equation}\label{eq:coh} 
W(f) \Omega = e^{-\| f\|^2 /2} \sum_{n \geq 0} \frac{ (a^* (f))^n}{n!} \Omega  =
e^{-\| f\|^2 /2} \sum_{n \geq 0} \frac{1}{\sqrt{n!}} \, \{ 0, \dots , 0, f^{\otimes n}, 0 \dots \} \,.
\end{equation}
Hence, if $P_n$ denotes the projection onto the $n$-particle sector $\cF^{(n)}$, we have, comparing with (\ref{eq:factor}),
\begin{equation}\label{eq:proj} P_n W(f) \Omega = e^{-\| f \|^2/2}  \frac{a^* (f)^n}{n!} \Omega = \frac{e^{-\| f \|^2/2}}{\sqrt{n!}} \,  \{ 0, \dots , 0, f^{\otimes n} , 0, \dots \} \, .\end{equation} 
We will use this formula with $n=N$ to write the initial product state $\ph^{\otimes N}$ as the projection onto the $N$-particle sector of an appropriately chosen coherent state. 
Eq. (\ref{eq:coh}) is a consequence of the expression 
\[ \exp (a^* (f) - a (f)) = e^{-\|f \|^2/2} \exp (a^* (f)) \exp
(-a(f)) \] which follows from the fact that $[a (f) , a^* (f)] = \| f \|^2$ commutes 
with $a(f)$ and $a^* (f)$. It implies that coherent states are
superpositions of states with different number of particles (the number of particles is a random variable with Poisson distribution having average $\| f \|^2$) . 
In the next lemma we list some key facts about Weyl operators and coherent states.
\begin{lem}\label{lm:coh}
Let $f,g \in L^2 (\bR^3, dx)$.
\begin{itemize}
\item[i)] The Weyl operator satisfy the relations
\[ W(f) W(g) = W(g) W(f) e^{-2i \, \text{Im} \, \langle f,g \rangle} = W(f+g) e^{-i\, \text{Im} \, \langle f,g \rangle} \,. \]
\item[ii)] $W(f)$ is a unitary operator and
\[ W(f)^* = W(f)^{-1}  = W (-f). \]
\item[iii)] We have \[ W^* (f) a_x W(f) = a_x + f(x), \qquad \text{and} \quad W^* (f) a^*_x
W(f) = a^*_x + \overline{f} (x) \, .\]
\item[iv)] {F}rom iii) we see that coherent states are eigenvectors of annihilation operators
\[ a_x \psi (f) = f(x) \psi (f)  \qquad \Rightarrow \qquad a (g)
\psi (f) = \langle g, f \rangle_{L^2} \psi (f) \, .\]
\item[v)] The expectation of the number of particles in the coherent
state $\psi (f)$ is given by $\| f\|^2$, that is
\[ \langle \psi (f), \cN \psi (f) \rangle = \| f \|^2
\, . \] Also the variance of the number of particles in $\psi (f)$
is given by $\|f \|^2$ (the distribution of $\cN$ is Poisson), that
is
\[ \langle \psi (f), \cN^2 \psi (f) \rangle - \langle \psi (f) ,
\cN \psi (f) \rangle^2 = \| f \|^2 \, .\]
\item[vi)] Coherent states are normalized but not orthogonal to each
other. In fact
\[ \langle \psi (f) , \psi (g) \rangle = e^{-\frac{1}{2}\left( \| f
\|^2 + \| g \|^2 - 2 (f,g) \right)}  \quad \Rightarrow \quad
|\langle \psi (f) , \psi (g) \rangle| = e^{-\frac{1}{2} \| f- g
\|^2} \, .\]
\end{itemize}
\end{lem}

\section{Proof of Proposition \ref{prop:main}}
\label{sec:proof}

Formulating the problem on the Fock space $\cF$, the one-particle density associated with the evolution $e^{-i \wt H_N t} \psi_N$ of the initial product state $\psi_N = \ph^{\otimes N}$ has the kernel  
\begin{equation}
\wt\gamma^{(1)}_{N,t} (y ; x) =  \frac{1}{N} \left\langle \frac{a^* (\ph)^N}{\sqrt{N!}} \Omega , e^{i\cH_N t} a_x^* a_y e^{-i \cH_N t} \frac{a^* (\ph)^N}{\sqrt{N!}} \Omega\right\rangle \, .
\end{equation}
Here we used (\ref{eq:factor}), (\ref{eq:margi}) and the fact that, on the $N$-particle sector, $\cH_N |_{\cF^{(N)}} = \wt H_N$. Now, using (\ref{eq:proj}), we write 
\begin{equation}\label{eq:PNW} 
\frac{a^* (\ph)^N}{\sqrt{N!}} \Omega =  \frac{\sqrt{N!}}{N^{N/2} e^{-N/2}} P_N W(\sqrt{N} \ph) \Omega = d_N \, P_N W(\sqrt{N} \ph) \Omega 
\end{equation}
where we defined 
\begin{equation}\label{eq:dN} d_N = \frac{\sqrt{N!}}{e^{-N/2} N^{N/2}} \leq C N^{1/4} \,.  \end{equation}
Thus, we obtain 
\begin{equation}\label{eq:gm1} \begin{split}
\wt \gamma^{(1)}_{N,t} (y;x) = \; &\frac{1}{N} \left\langle \frac{a^* (\ph)^N}{\sqrt{N!}} \Omega , e^{i\cH_N t} a_x^* a_y e^{-i \cH_N t} \frac{a^* (\ph)^N}{\sqrt{N!}} \Omega\right\rangle  \\ = \; &\frac{d_N}{N} \left\langle \frac{a^* (\ph)^N}{\sqrt{N!}} \Omega , e^{i\cH_N t} a_x^* a_y e^{-i \cH_N t} P_N W (\sqrt{N}\ph) \Omega\right\rangle  \\
 = \; &\frac{d_N}{N} \left\langle \frac{a^* (\ph)^N}{\sqrt{N!}} \Omega , P_N e^{i\cH_N t} a_x^* a_y e^{-i \cH_N t} W (\sqrt{N}\ph) \Omega\right\rangle \\
   = \; &\frac{d_N}{N} \left\langle \frac{a^* (\ph)^N}{\sqrt{N!}} \Omega , e^{i\cH_N t} a_x^* a_y e^{-i \cH_N t} W (\sqrt{N}\ph) \Omega\right\rangle \, .
\end{split}
\end{equation}

Following, similarly to \cite{RS}, an idea 
first introduced by Hepp in \cite{He}, we define the unitary evolution 
\begin{equation}\label{eq:U} \cU (t;s) = W^* (\sqrt{N} \wt \ph_t) e^{-i (t-s) \cH_N} W(\sqrt{N} \wt \ph_s) \, . \end{equation}
Lemma \ref{lm:coh} implies that 
\begin{equation}\label{eq:hepp} 
\cU^* (t;0) \, a_y \, \cU (t;0) = W^* (\sqrt{N} \ph) e^{i\cH_N t} (a_y - \sqrt{N} \wt \ph_t (y)) e^{-i\cH_N t} W (\sqrt{N} \ph) \, . \end{equation}
Moreover, $\cU (t;s)$ satisfies the Schr\"odinger equation 
\[ i \frac{d}{dt} \cU (t;s) = \cL (t) \cU (t;s), \qquad \text{with } \cU (s;s) = 1 \]
with the time dependent generator 
\begin{equation}\label{eq:cL} \begin{split} 
\cL (t) = \; & \int dx \, \nabla_x a^*_x \nabla_x a_x + \int dx \, (\wt V*|\wt\ph_t|^2 ) (x)\, a^*_x a_x \\ &+ \int dx dy \, \wt V(x-y) a_x^* a_y \wt\ph_t (x) \overline{\wt\ph}_t (y) + \int dx dy \, \wt V(x-y) \left( a_x^* a_y^* \, \wt \ph_t (x) \wt \ph_t (y) +  a_x a_y \, \overline{\wt \ph}_t (x) \overline{\wt \ph}_t (y) \right) \\ &+\frac{1}{\sqrt{N}} \int dx dy \, \wt V(x-y) a_x^* \left( a_y^* \wt \ph_t (y) + a_y \overline{\wt \ph}_t (y) \right) a_x \\ &+\frac{1}{N} \int dx dy \wt V(x-y) a_x^* a_y^* a_y a_x  \, .\end{split} \end{equation}
{F}rom (\ref{eq:gm1}) and (\ref{eq:hepp}), we obtain 
\begin{equation}\begin{split}
\wt \gamma^{(1)}_{N,t} (y;x) = \; &\frac{d_N}{N} \left\langle \frac{a^* (\ph)^N}{\sqrt{N!}} \Omega , W (\sqrt{N} \ph) \, \cU^* (t;0) (a_x^* + \sqrt{N} \, \overline{\wt\ph}_t (x)) (a_y + \sqrt{N} \wt \ph_t (y)) \cU(t;0) \Omega\right\rangle \\ =\; &  \wt \ph_t (y) \, \overline{\wt \ph_t} (x)+ \frac{d_N}{N} \left\langle \frac{a^* (\ph)^N}{\sqrt{N!}} \Omega , W (\sqrt{N} \ph) \cU^* (t;0) a_x^* a_y \cU(t;0) \Omega\right\rangle \\ &+ \frac{d_N}{\sqrt{N}} \, \overline{\wt \ph_t} (x) \left\langle \frac{a^* (\ph)^N}{\sqrt{N!}} \Omega , W (\sqrt{N} \ph) \cU^* (t;0) a_y \cU(t;0) \Omega\right\rangle
 \\ &+ \frac{d_N}{\sqrt{N}} \, \wt \ph_t (y) \left\langle \frac{a^* (\ph)^N}{\sqrt{N!}} \Omega , W (\sqrt{N} \ph) \cU^* (t;0) a^*_x \cU(t;0) \Omega\right\rangle \, .
\end{split}
\end{equation}
Integrating against the kernel $J (x;y)$ of a compact hermitian one-particle 
operator $J$ (hermiticity implies that $J(y;x) = \overline{J}(x;y)$), we find 
\begin{equation}\label{eq:diff1}
\begin{split}
\tr \, J \left( \wt\gamma^{(1)}_{N,t}  - |\wt \ph_t \rangle \langle \wt \ph_t| \right) = \; & \int dx dy \, J(x;y) \left(\wt \gamma^{(1)}_{N,t} (y;x) - \wt \ph_t (y) \, \overline{\wt \ph_t} (x) \right)  \\ =\; & \frac{d_N}{N} \left\langle \frac{a^* (\ph)^N}{\sqrt{N!}} \Omega , W (\sqrt{N} \ph) \cU^* (t;0) d\Gamma (J) \cU(t;0) \Omega\right\rangle \\ &+ \frac{d_N}{\sqrt{N}} \left\langle \frac{a^* (\ph)^N}{\sqrt{N!}} \Omega , W (\sqrt{N} \ph) \cU^* (t;0) \phi (J \wt\ph_t) \cU(t;0) \Omega\right\rangle 
\end{split}
\end{equation}
where we used the definition (\ref{eq:phi}). To bound the second term on the r.h.s. of (\ref{eq:diff1}), we write 
\[ \begin{split} \frac{d_N}{\sqrt{N}} \Big\langle \frac{a^* (\ph)^N}{\sqrt{N!}} \Omega , &W (\sqrt{N} \ph) \cU^* (t;0) \phi (J \wt\ph_t) \cU(t;0) \Omega\Big\rangle  \\ 
= &\;  \frac{d_N}{\sqrt{N}} \left\langle \frac{a^* (\ph)^N}{\sqrt{N!}} \Omega , W (\sqrt{N} \ph) \, \cU_2^* (t;0) \phi (J \wt\ph_t) \cU_2 (t;0) \Omega\right\rangle \\
&+ \frac{d_N}{\sqrt{N}} \left\langle \frac{a^* (\ph)^N}{\sqrt{N!}} \Omega , W (\sqrt{N} \ph) \, \left( \cU^* (t;0) - \cU_2^* (t;0) \right) \, \phi (J \wt\ph_t) \cU_2 (t;0) \Omega\right\rangle 
\\ &+ \frac{d_N}{\sqrt{N}} \left\langle \frac{a^* (\ph)^N}{\sqrt{N!}} \Omega , W (\sqrt{N} \ph) \, \cU^* (t;0) \, \phi (J \wt\ph_t) \left( \cU (t;0) -  \cU_2 (t;0)\right)  \Omega\right\rangle
\end{split}\]
where we compared the fluctuation dynamics $\cU (t;0)$ with the dynamics $\cU_2 (t;0)$ defined by the equation
\begin{equation}\label{eq:U2} i \frac{d}{dt} \cU_2 (t;s) = \cL_2 (t) \cU_2 (t;s), \qquad \text{with } \quad \cU_2 (s;s) = 1 \end{equation}
and with the quadratic generator 
\begin{equation}\label{eq:L2} \begin{split} 
\cL_2 (t) = \; & \int dx \nabla_x a^*_x \nabla_x a_x + \int dx \,  (\wt V*|\wt \ph_t|^2 ) (x) a^*_x a_x \\ &+ \int dx dy \, \wt V(x-y) a_x^* a_y \wt \ph_t (x) \overline{\wt \ph}_t (y) \\ &+ \int dx dy \, \wt V(x-y) \left( a_x^* a_y^* \, \wt \ph_t (x) \wt \ph_t (y) +  a_x a_y \, \overline{\wt \ph}_t (x) \overline{\wt \ph}_t (y)\right) \,. \end{split}  \end{equation}
The existence of the evolution $\cU_2 (t;s)$ has been established in \cite{GV}. Taking absolute value in (\ref{eq:diff1}), we find 
\begin{equation}\label{eq:diff2}
\begin{split}
\Big| \tr \, J \, \Big( \wt\gamma^{(1)}_{N,t}  - &|\wt \ph_t \rangle \langle\wt \ph_t| \Big)  \Big| \\ \leq \; & 
\frac{d_N}{N} \left\| (\cN+1)^{-1/2} W^* (\sqrt{N} \ph) \frac{a^* (\ph)^N}{\sqrt{N!}} \Omega \right\| \, \left\| (\cN+1)^{1/2} \, \cU^* (t;0) d\Gamma (J) \cU(t;0) \Omega \right\| \\ &+ \frac{d_N}{\sqrt{N}} \left| \left\langle W^* (\sqrt{N} \ph)  \frac{a^* (\ph)^N}{\sqrt{N!}} \Omega , \cU_2^* (t;0) \phi (J\wt \ph_t) \cU_2 (t;0) \Omega\right\rangle \right| \\
&+ \frac{d_N}{\sqrt{N}} \left| \left\langle W^* (\sqrt{N} \ph)  \frac{a^* (\ph)^N}{\sqrt{N!}} \Omega , (\cU_2^* (t;0) - \cU^* (t;0)) \, \phi (J\wt \ph_t) \cU_2 (t;0) \Omega\right\rangle \right| \\
&+ \frac{d_N}{\sqrt{N}} \left| \left\langle W^* (\sqrt{N} \ph)  \frac{a^* (\ph)^N}{\sqrt{N!}} \Omega , \cU^* (t;0) \, \phi (J\wt \ph_t) (\cU_2 (t;0) -\cU (t;0) )\Omega\right\rangle \right| 
\, .
\end{split}
\end{equation}

The first term on the r.h.s. can be bounded, using Lemma \ref{lm:2}, Proposition \ref{lm:1}, and Lemma \ref{lm:0}, by
\[ \begin{split} \frac{d_N}{N} \left\| (\cN+1)^{-1/2} W^* (\sqrt{N} \ph) \frac{a^* (\ph)^N}{\sqrt{N!}} \Omega \right\| \,& \left\| (\cN+1)^{1/2} \, \cU^* (t;0) d\Gamma (J) \cU(t;0) \Omega \right\| 
\\ \leq \; & \frac{C}{N} e^{\wt{K} |t|} \left\| (\cN+1)^2 \, d\Gamma (J) \cU(t;0) \Omega \right\| 
\\ \leq \; & \frac{C \| J \|}{N} e^{\wt{K} |t|} \left\| (\cN+1)^3 \,  \cU(t;0) \Omega \right\| \\
\leq \; & \frac{C \| J \|}{N} e^{2 \wt{K} |t|} \, \| ( \cN+1)^{9/2}\,  \Omega \| \leq \frac{C \| J \|}{N} e^{K |t|} \, . \end{split} \]

The second term on the r.h.s. of (\ref{eq:diff2}) vanishes. This follows from Lemma \ref{lm:3} and because 
\[ \begin{split} P_1 W^*(\sqrt{N} \ph) a^* (\ph)^N \Omega = \; & P_1 (a^* (\ph) + \sqrt{N})^N W^* (\sqrt{N} \ph) \Omega \\ = \; & e^{-N/2} P_1 (a^* (\ph) + \sqrt{N})^N (\Omega - \sqrt{N} a^* (\ph) \Omega) \\ = \; & -e^{-N/2} N^{N/2} \sqrt{N} a^* (\ph) \Omega + e^{-N/2} N^{\frac{N-1}{2}} N a^* (\ph) \Omega = 0 \, . \end{split}\]

The third term on the r.h.s. of (\ref{eq:diff2}) is bounded, from Proposition \ref{prop:comp2}, by
\[ \frac{d_N}{\sqrt{N}} \left| \left\langle W^* (\sqrt{N} \ph)  \frac{a^* (\ph)^N}{\sqrt{N!}} \Omega , (\cU_2^* (t;0) - \cU^* (t;0)) \, \phi (J\wt \ph_t) \cU_2 (t;0) \Omega\right\rangle \right| \leq \frac{C \| J \| e^{K |t|}}{N} \, . \]

The fourth and last term on the r.h.s. of (\ref{eq:diff2}) is bounded, using Lemma \ref{lm:2} and Lemma \ref{lm:1}, 
\[ \begin{split}
 \frac{d_N}{\sqrt{N}} &\left| \left\langle W^* (\sqrt{N} \ph)  \frac{a^* (\ph)^N}{\sqrt{N!}} \Omega , \cU^* (t;0) \, \phi (J\wt \ph_t) (\cU_2 (t;0) -\cU (t;0) )\Omega\right\rangle \right| 
\\ \leq\; & \frac{d_N}{\sqrt{N}} \left\| (\cN+1)^{-1/2}  W^* (\sqrt{N} \ph)  \frac{a^* (\ph)^N}{\sqrt{N!}} \Omega \right\| \, \left\| (\cN+1)^{1/2} \cU^* (t;0) \, \phi (J\wt \ph_t) (\cU_2 (t;0) -\cU (t;0) )\Omega \right\| \\ \leq\; & \frac{Ce^{\wt{K}|t|}}{\sqrt{N}} \left\| (\cN+1)^{2} \phi (J\wt \ph_t) (\cU_2 (t;0) -\cU (t;0) )\Omega \right\|
\end{split} \]
Writing $(\cN+1)^2 \phi (J\ph_t) = a^* (J \ph_t) (\cN+2)^2 + a (J \ph_t) \cN^2$, and using Lemma \ref{lm:0}, we find
\[  \begin{split} \frac{d_N}{\sqrt{N}} \Big| \Big\langle W^* (\sqrt{N} \ph)   \frac{a^* (\ph)^N}{\sqrt{N!}} \Omega , \cU^* (t;0) \,& \phi (J\wt \ph_t) (\cU_2 (t;0) -\cU (t;0) )\Omega\Big\rangle \Big| 
\\ \leq\; & \frac{C \| J \| \, e^{\wt{K} |t|}}{\sqrt{N}} \left\| (\cN+1)^{5/2} \, (\cU_2 (t;0) -\cU (t;0) )\Omega \right\| \\
\leq\; & \frac{C \| J \| \, e^{K |t|}}{N} \end{split} \]
where, in the last step, we used Proposition \ref{lm:4}. 

Summarizing, we showed that
\[ \Big|\tr \, J \left( \wt\gamma^{(1)}_{N,t}  - |\wt \ph_t \rangle \langle \wt \ph_t| \right)  \Big| \leq  \frac{C \| J \|}{N} e^{K |t|}\] 
for all compact hermitian operators $J$ on $L^2 (\bR^3, dx)$. Since the space of compact operators is the dual to the trace class operators, and since $\wt \gamma^{(1)}_{N,t}$ and $|\wt \ph_t \rangle \langle \wt \ph_t|$ are hermitian, we immediately obtain that  
\[ \tr \, \left| \wt \gamma^{(1)}_{N,t} - |\wt \ph_t \rangle \langle \wt \ph_t| \right| \leq  \frac{C}{N} e^{K |t|} \]
which concludes the proof of Proposition \ref{prop:main}.

\section{Bounds on the growth of number of particles}
\label{sec:bdN}

One of the most important ingredients in the proof of Theorem \ref{thm:main} presented in the previous section is a bound on the growth of the number of particles with respect to the evolutions $\cU (t;s)$ and $\cU_2 (t;s)$. 

\begin{prop}\label{lm:1}
Suppose that $\cU (t;s)$ and $\cU_2 (t;s)$ are the unitary evolutions defined in (\ref{eq:U}) and (\ref{eq:U2}), respectively. Then, for every $j \in \bN$, there exist constants $C_j, K_j >0$ such that 
\[ \| (\cN+1)^j \cU_2 (t;s) \psi \| \leq C_j \, e^{K_j |t-s|} \| (\cN +1)^j \psi \| \]
and 
\[  \| (\cN+1)^j \cU (t;s) \psi \| \leq C_j \, e^{K_j |t-s|} \| (\cN+1)^{2j+1} \psi \| \]
for every $\psi \in \cF$, $t \in \bR$. This implies that, for any $j \in \bN$, the operators $(\cN+1)^j \cU_2 (t;s) (\cN+1)^{-j}$ and $(\cN+1)^j \cU (t;s) (\cN+1)^{-2j-1}$ extend as bounded operators on the Fock space $\cF$ with norm bounded by 
\[ 
\| (\cN+1)^j \, \cU_2 (t;s)\, (\cN+1)^{-j} \| \leq C_j \, e^{K_j |t-s|} \] 
and \[
\| (\cN+1)^j  \, \cU (t;s) \, (\cN+1)^{-2j-1} \| \leq C_j \, e^{K_j |t-s|} \,. \]
\end{prop}

The proof of this proposition can be found in \cite{RS}. More precisely, the bound for the dynamics $\cU (t;s)$ is given in Proposition 3.3 of \cite{RS}. On the other hand, the bound for the dynamics $\cU_2 (t;s)$ (which is much simpler), can be obtained using arguments very similar to those of Lemma 3.5 of \cite{RS} (where a different cutoffed dynamics is studied).

\section{Comparison of Dynamics}
\label{sec:compare}

The goal of this section is to estimate the difference between the full fluctuation evolution $\cU (t;s)$ and the dynamics $\cU_2 (t;s)$. 
\begin{prop}\label{lm:4}
Suppose that, in the definition (\ref{eq:wtV}) of the regularized potential $\wt V$, the cutoff $\alpha_N$ is such that $\alpha_N \geq N^{-r}$, for some $r \in \bN$. 
Suppose that $\cU (t;s)$ and $\cU_2 (t;s)$ are the unitary evolutions defined in (\ref{eq:U}) and (\ref{eq:U2}), respectively. Then, for any $j \in \bN$, there exist constants $C_{j,r}, K_{j,r} >0$ such that 
 \[ \| (\cN+1)^j \, \left( \cU (t;s) - \cU_2 (t;s)\right) \psi \| \leq \frac{C_{j,r} e^{K_{j,r} |t-s|}}{\sqrt{N}} \, \langle \psi, (\cK + \cN^{4r (2j+3)} + 1) \psi \rangle^{1/2},  \] where $\cK$ is the kinetic energy operator 
 \begin{equation}\label{eq:cK} \cK = d\Gamma (-\Delta) = \int dx \, \nabla_x a_x^* \nabla_x a_x \, . \end{equation}
 \end{prop}
 
\begin{proof}
We fix $t \geq 0$ and $s =0$ (all other cases can be treated analogously). Using 
\[  \cU (t;0) - \cU_2 (t;0) = \int_0^t d\tau \, \cU (t;\tau) \left( \cL (\tau) - \cL_2 (\tau) \right) \cU_2 (\tau;0) \]
we find that 
\begin{equation}\label{eq:L3+L4} \begin{split} 
\| (\cN+1)^j \, \left( \cU (t;0) - \cU_2 (t;0) \right) \psi \| \leq \; &\int_0^t d\tau \, \| (\cN+1)^j \cU (t;\tau) \left( \cL (\tau) - \cL_2 (\tau) \right) \cU_2 (\tau;0) \psi \| \\ \leq \; &\int_0^t d\tau \, e^{K (t-\tau)} 
\| (\cN+1)^{2j+1} \cL_3 (\tau) \,  \cU_2 (\tau;0) \psi \|  
\\ & + \int_0^t d\tau \, e^{K (t-\tau)} \| (\cN+1)^{2j+1} \cL_4 \, \cU_2 (\tau;0) \psi \|  
\end{split} \end{equation}
where we used Proposition \ref{lm:1} and we wrote $\cL (\tau) -\cL_2 (\tau) = \cL_3 (\tau) + \cL_4 $, with 
\begin{equation}\label{eq:defL3L4}
\begin{split} \cL_3 (t) &= \frac{1}{\sqrt{N}} \int dx dy \, \wt V(x-y) \, a_x^* (a_y^* \wt \ph_t (y) + a_y \overline{\wt \ph}_t (y)) a_x \\   \cL_4 &= \frac{1}{N} \int dx dy \, \wt V(x-y) \, a_x^* a_y^* a_y a_x  \, . \end{split} \end{equation}
Using Lemma \ref{lm:L3} to bound the first term on the r.h.s. of (\ref{eq:L3+L4}), we find 
\begin{equation}\label{eq:NH4N} \begin{split} 
\| (\cN+1)^j \,& \left( \cU (t;0) - \cU_2 (t;0)\right) \psi \|  \\ \leq \; & \frac{C}{\sqrt{N}} e^{K t} \| (\cN+1)^{2j+5/2} \psi \| \\ & + \int_0^t d\tau \, e^{K(t-\tau)} \, \langle \cU_2 (\tau;0) \psi,  (\cN+1)^{2j+1} \cL^2_4 (\cN+1)^{2j+1} \cU_2 (\tau;0) \psi \rangle^{1/2} \, . \end{split} \end{equation}
To bound the second term on the r.h.s. of the last equation we observe that, restricting the operators on the $n$-particle sector $\cF^{(n)}$ of the Fock space, 
\[ \begin{split}
(\cN+1)^{2j+1} \,& \cL_4^2 \, (\cN+1)^{2j+1} |_{\cF^{(n)}} \\ = \; &\frac{(n+1)^{4j+2}}{N^2} \left( \sum_{i<j}^n \wt V (x_i - x_j) \right)^2 \leq \frac{(n+1)^{4j+4}}{N^2} \sum_{i<j}^n \wt V^2 (x_i - x_j) \\ \leq \; & \frac{{\bf 1} (n+1 \leq N^{\frac{1}{4j+5}})}{N^2} (n+1)^{4j+5} \sum_{j=1}^n (-\Delta_{x_j} + 1) + \frac{{\bf 1} (n+1 \geq N^{\frac{1}{4j+5}})}{N^2 \alpha_N^2} (n+1)^{4j+6} \\ \leq \; &
\left[ \frac{1}{N} (\cK + \cN) + \frac{{\bf 1} (\cN +1 \geq N^{\frac{1}{4j+5}})}{N^2 \alpha_N^2} \, (\cN+1)^{4j+6} \right] |_{\cF^{(n)}} \end{split} \]
where we used the bounds $\wt V^2 (x) \leq 1-\Delta$ and $|\wt V(x)| \leq \alpha^{-1}_N$ (see (\ref{eq:wtV})) and $\cK$ is defined in (\ref{eq:cK}). In Lemma \ref{lm:7} below we show that there exists a constant $C >0$ such that, for arbitrary $\tau\in \bR$,  
\[ \cK \leq \cL_2 (\tau) + C ( \cN +1) \,. \] 
Hence 
\[ (\cN+1)^{2j+1} \, \cL_4^2 \, (\cN+1)^{2j+1} \leq \frac{1}{N} (\cL_2 (\tau) + C (\cN+1)) + \frac{{\bf 1} (\cN +1 \geq N^{\frac{1}{4j+5}})}{N^2 \alpha_N^2} (\cN+1)^{4j+6} \]
and therefore
\begin{equation} \label{eq:H2+N} \begin{split} 
\langle \cU_2 (\tau;0) \psi,  (\cN+1)^{2j+1} &\, \cL^2_4 \, (\cN+1)^{2j+1} \cU_2 (\tau;0) \psi \rangle \\ \leq \; & \frac{1}{N} \, \langle \cU_2 (\tau;0) \psi, (\cL_2 (\tau) + C ( \cN +1)) \cU_2 (\tau;0) \psi \rangle \\ &+ \frac{1}{N^2 \alpha_N^2}  \langle \cU_2 (\tau;0) \psi, {\bf 1} (\cN +1 \geq N^{\frac{1}{4j+5}}) (\cN+1)^{4j+6} \cU_2 (\tau;0) \psi \rangle \, . 
\end{split}\end{equation}

Using the bound ${\bf 1} (x \geq 1) \leq x^{m(4j+5)}$, valid for every $m \in \bN$, Proposition \ref{lm:1} and the assumption $\alpha_N \geq N^{-r}$, we can estimate the second term on the r.h.s. of the last equation by  \[ \begin{split} \frac{1}{N^2 \alpha_N^2}  \langle \cU_2 (\tau;0) \psi, &{\bf 1} (\cN +1 \geq N^{\frac{1}{4j+5}}) (\cN+1)^{4j+6} \cU_2 (\tau;0) \psi \rangle \\ 
& \leq \frac{1}{N^2 \alpha_N^2} \frac{1}{N^m} \langle \cU_2 (\tau;0) \psi,  (\cN+1)^{(4j+6)(m+1)} \cU_2 (\tau;0) \psi \rangle \\ & \leq \frac{C}{N^{m+2-2r}} e^{K \tau}  \| (\cN+1)^{(2j+3) (m+1)} \psi \|^2 
\end{split} \]
for appropriate constants $C,K$ (depending on $m$ and $j$). Fixing $m = 2r - 1$, we find
\[  \frac{1}{N^2 \alpha_N^2}  \langle \cU_2 (\tau;0) \psi, {\bf 1} (\cN +1 \geq N^{\frac{1}{4j+5}}) (\cN+1)^{4j+6} \cU_2 (\tau;0) \psi \rangle \leq \frac{C}{N} e^{K \tau} \, \| (\cN+1)^{2r (2j+3)} \psi \|^2 
\,.  \]
To control the first term on the r.h.s. of (\ref{eq:H2+N}) we use Lemma \ref{lm:L2} and Proposition \ref{lm:1} together with the fact that, at $\tau=0$, 
\[ \langle \psi, (\cL_2 (0) + C ( \cN +1) ) \psi \rangle \leq C \langle \psi, (\cK + \cN + 1) \psi \rangle \,. \]
We conclude that there exist constants $C,K >0$, depending on $r$ and $j$, such that 
\[  \langle \cU_2 (\tau;0) \psi,  (\cN+1)^{2j+1} \, \cL^2_4 \, (\cN+1)^{2j+1} \cU_2 (\tau;0) \psi \rangle  \leq \frac{C}{N} e^{K \tau} \, \langle \psi, (\cK + \cN^{4r (2j+3)} + 1) \psi \rangle . \]
Inserting this bound in (\ref{eq:NH4N}), we obtain the desired estimate.
\end{proof}

We also need to bound the difference between the two evolutions $\cU (t;0)$ and $\cU_2 (t;0)$ in the third term on the r.h.s. of (\ref{eq:diff2}). This is the content of the next proposition.

\begin{prop}\label{prop:comp2}
Suppose that, in the definition (\ref{eq:wtV}) of the regularized potential $\wt V$, the cutoff $\alpha_N$ is such that $\alpha_N \geq N^{-r}$, for some $r \in \bN$. Let $\wt{\ph}_t$ be the solution of the regularized Hartree equation (\ref{eq:hartreereg}). Suppose that $\cU (t;s)$ and $\cU_2 (t;s)$ are the unitary evolutions defined in (\ref{eq:U}) and (\ref{eq:U2}), respectively. Suppose that $J$ is a bounded hermitian operator on $L^2 (\bR^3, dx)$. Then, there exist constants $C,K >0$ such that
\[ \left| \left\langle \frac{a^* (\ph)^N}{\sqrt{N!}} \Omega, W(\sqrt{N} \ph) \, (\cU^* (t;0) - \cU^*_2 (t;0) ) \, \phi (J \wt \ph_t) \,\cU_2 (t;0)  \Omega \right\rangle \right| \leq \frac{C \| J \| e^{K |t|}}{d_N \, \sqrt{N}} \] 
where $d_N = \sqrt{N!} / (e^{-N/2} N^{N/2}) \leq C N^{1/4}$ was defined in (\ref{eq:dN}).
\end{prop}

\begin{proof} 
We fix $t \geq 0$ and we write
\[ \begin{split}
\Big\langle \frac{a^* (\ph)^N}{\sqrt{N!}} \Omega, & W(\sqrt{N} \ph)  \, \left(\cU^* (t;0) - \cU^*_2 (t;0) \right)  \phi (J \wt \ph_t) \cU_2 (t;0) \Omega \Big\rangle \\ = \; & \int_0^t ds \, 
\Big\langle \frac{a^* (\ph)^N}{\sqrt{N!}} \Omega, W(\sqrt{N} \ph) \, \cU_2^* (t;s) \left( \cL_3 (s) + \cL_4 \right) \cU^* (s;0)   \phi (J \wt \ph_t) \cU_2 (t;0) \Omega \Big\rangle
\end{split} \]
where $\cL_3 (s)$ and $\cL_4$ are defined in (\ref{eq:defL3L4}). Taking the absolute value, we find
\begin{equation}\label{eq:AB} \begin{split} 
\Big| \Big\langle \frac{a^* (\ph)^N}{\sqrt{N!}} \Omega, W& (\sqrt{N} \ph) \, (\cU^* (t;0) - \cU^*_2 (t;0) ) \, \phi (J \wt \ph_t) \,\cU_2 (t;0)  \Omega \Big\rangle \Big|
 \\ \leq \; & \int_0^t ds \, \left\| (\cN+1)^{-1/2} \, W^* (\sqrt{N} \ph) \frac{a^* (\ph)^N}{\sqrt{N!}} \Omega \right\| \\ &\hspace{3cm} \times  \left\| (\cN+1)^{1/2} \, \cU^*_2 (t;s) \cL_3 (s) \cU^* (s;0) \phi (J \wt \ph_t) \cU_2 (t;0) \Omega \right\| \\
&+ \int_0^t ds \, \left\|  (\cN+1)^{-\kappa}  \, \cL_4 \, \cU_2 (t;s) W^* (\sqrt{N} \ph) 
\frac{a^* (\ph)^N}{\sqrt{N!}} \Omega \right\| \\ &\hspace{3cm} \times \Big\| (\cN+1)^\kappa \cU^* (s;0) \phi (J \wt \ph_t) \cU_2 (t;0) \Omega \Big\| \\
=: \, & \text{A} + \text{B} \,  
\end{split} \end{equation}
where the parameter $\kappa >0$ will be fixed later on. To bound the term $\text{A}$ we note that, by Lemma \ref{lm:2}, Lemma \ref{lm:1}, Lemma \ref{lm:L3}, we have
\[ \begin{split} 
\text{A} \leq \; &\frac{C}{d_N} \int_0^t ds \, e^{\wt{K} (t-s)} \left\| (\cN+1)^{1/2} \cL_3 (s) \cU^* (s;0) \phi (J \wt \ph_t) \cU_2 (t;0) \Omega \right\| \\ \leq \; &\frac{C}{\sqrt{N} d_N} \int_0^t ds \, e^{\wt{K} (t-s)} \left\| (\cN+1)^{2} \, \cU^* (s;0) \phi (J \wt \ph_t) \cU_2 (t;0) \Omega \right\| \\ \leq \; &\frac{C}{\sqrt{N} d_N} \int_0^t ds \, e^{\wt{K} t} \left\| (\cN+1)^{5} \phi (J \wt \ph_t) \cU_2 (t;0) \Omega \right\| \,.
\end{split} \]
Writing $(\cN+1)^5 \, \phi (J \wt \ph_t) = a^* (J \wt \ph_t) (\cN+2)^5 + a (J \wt \ph_t) \cN^5$, and using Lemma \ref{lm:0}, we conclude (again by Lemma \ref{lm:1}) that
\begin{equation}\label{eq:A-pr} \begin{split} 
 \text{A} \leq \; & \frac{C \| J \|}{\sqrt{N} d_N} \int_0^t ds \, e^{\wt{K} t} \left\| (\cN+1)^{11/2} \,\cU_2 (t;0) \Omega \right\|  \leq  \frac{C \| J \| e^{K t}}{\sqrt{N} d_N}
\end{split} \end{equation}
for an appropriate constant $K >0$. 

Next, we estimate the term $\text{B}$ on the r.h.s. of (\ref{eq:AB}). On the one hand, we have 
\begin{equation}\label{eq:B0} \begin{split} \left\| (\cN+1)^\kappa \cU^* (s;0) \phi (J \wt \ph_t) \cU_2 (t;0) \Omega \right\| \leq \; &C e^{\wt{K} s} \left\| (\cN+1)^{2\kappa+1} \phi (J \wt \ph_t) \cU_2 (t;0) \Omega \right\| \\\leq \; &C e^{\wt{K} s} \| J \| \, \left\| (\cN+1)^{2\kappa +3/2}  \, \cU_2 (t;0) \Omega \right\| \leq C  \| J \| \, e^{\wt{K} (s+t)} \end{split} \end{equation}
where the constants $C, \wt{K}$ depend on $\kappa >0$. To bound the other norm in the term $\text{B}$, we write
\[ W^* (\sqrt{N} \ph ) \frac{a^* (\ph)^N}{\sqrt{N!}} \Omega = \sum_{m\geq 0} \cA_m \, \frac{a^* (\ph)^m}{\sqrt{m!}} \Omega\,. \] 
Note that, by the unitarity of the Weyl operators (recall that $a^* (\ph)^{m} \Omega = \sqrt{m!} \, \{ 0, \dots, 0 , \ph^{\otimes m}, 0, \dots \}$)
\[ \sum_{m \geq 1} |\cA_m|^2  = 1 \, . \]
Hence
\begin{equation}\label{eq:B2} \begin{split} \Big\| (\cN+1)^{-\kappa} \, \cL_4 \, &\cU_2 (t;s) W^* (\sqrt{N} \ph) 
\frac{a^* (\ph)^N}{\sqrt{N!}} \Omega \Big\|^2 \\ = \; & \sum_{m,\ell \geq 0} \cA_m \, \overline{\cA}_\ell \, \left\langle \frac{a^* (\ph)^\ell}{\sqrt{\ell !}}\Omega, \cU_2^* (t;s) \,  (\cN+1)^{-\kappa} \, \cL_4^2  \,  (\cN+1)^{-\kappa}\, \cU_2 (t;s) \frac{a^* (\ph)^m}{\sqrt{m!}} \Omega \right\rangle \\ \leq \; & \sum_{m,\ell \geq 0} |\cA_m|^2 \left\langle \frac{a^* (\ph)^\ell}{\sqrt{\ell !}} \Omega, \cU_2^* (t;s) \,  (\cN+1)^{-\kappa}\,  \cL_4^2  \,  (\cN+1)^{-\kappa}\,\cU_2 (t;s) \frac{a^* (\ph)^\ell}{\sqrt{\ell !}} \Omega \right\rangle \\ \leq \; & \frac{1}{N^2 \alpha_N^2} \sum_{\ell \geq N^{\delta}} \left\| (\cN+1)^{-\kappa+2} \,  \cU_2 (t;s) \frac{a^* (\ph)^\ell}{\sqrt{\ell !}} \Omega \right\|^2 \\ &+ \sum_{\ell < N^{\delta}} 
 \left\langle \frac{a^* (\ph)^\ell}{\sqrt{\ell !}} \Omega , \cU_2^* (t;s) \, \cL_4^2  \, \cU_2 (t;s) \frac{a^* (\ph)^\ell}{\sqrt{\ell !}}\Omega \right\rangle \end{split} \end{equation}
where $0 < \delta < 1$ will be fixed later on. In the regime $\ell \geq N^{\delta}$, we estimated \[ \cL_4^2 \leq \frac{1}{N^2 \alpha_N^2} (\cN+1)^4 \, . \] In the regime $\ell < N^\delta$, on the other hand, we used the fact that $\cL_4$ commutes with $\cN$, and that $(\cN+1)^{-2\kappa} \leq 1$. By Lemma \ref{lm:1}, we have
\[ \begin{split} \left\|  (\cN+1)^{-\kappa+2}  \, \cU_2 (t;s) \frac{a^* (\ph)^\ell}{\sqrt{\ell !}} \Omega \right\| \leq \; & \left\|  (\cN+1)^{-\kappa+2}  \, \cU_2 (t;s)  (\cN+1)^{\kappa-2} \right\| \,  \left\| (\cN+1)^{-\kappa+2}  \, \frac{a^* (\ph)^\ell}{\sqrt{\ell !}} \Omega \right\|
\\ \leq\; & C e^{\wt{K} (t-s)/2} \frac{1}{(\ell+1)^{\kappa-2}}  \end{split} \]
where the constants $C,\wt{K}$ depend on $\kappa$. Therefore,
\begin{equation}\label{eq:B21} \frac{1}{N^2 \alpha_N^2} \sum_{\ell \geq N^{\delta}} \left\| (\cN+1)^{-\kappa+2} \,  \cU_2 (t;s) \frac{a^* (\ph)^\ell}{\sqrt{\ell !}} \Omega \right\|^2 \leq \frac{C e^{\wt{K} (t-s)}}{N^2 \alpha_N^2} \sum_{\ell \geq N^\delta} \frac{1}{(\ell+1)^{2\kappa-4}} \leq \frac{C e^{\wt{K} (t-s)}}{N^{2+\delta (2\kappa-6)} \alpha_N^2} \,.\end{equation}

To bound the second term on the r.h.s. of (\ref{eq:B2}), we observe that, on the $n$-particle sector,
\[ \begin{split} \cL_4^2 |_{\cF_n} = \; & \frac{1}{N^2} \left(\sum_{i<j}^n \wt{V} (x_i - x_j) \right)^2 \leq   \frac{n^2}{N^2} \sum_{i<j}^n \wt{V}^2 (x_i -x_j) \\ \leq \; &
\frac{n^3 {\bf 1} (n \leq N^{1/12})}{N^2} \sum_{j=1}^n (1-\Delta_{x_j}) + \frac{n^4 {\bf 1} (n \geq N^{1/12})}{N^2 \alpha_N^2} \,.
\end{split} \]
This implies that 
\[  \cL_4^2 \leq  \frac{1}{N^{7/4}} (\cK+\cN) + \frac{\cN^4 {\bf 1} (\cN \geq N^{1/12})}{N^2 \alpha_N^2} \, , \]
where $\cK$ is the kinetic energy operator defined in (\ref{eq:cK}). We find
\[ \begin{split}
\sum_{\ell < N^\delta} \Big\langle  \frac{a^* (\ph)^\ell}{\sqrt{\ell !}} \Omega, \, & \cU_2^* (t;s) \, \cL_4^2 \, \cU_2 (t;s)  \frac{a^* (\ph)^\ell}{\sqrt{\ell !}} \Omega \Big\rangle \\ \leq \; & \frac{1}{N^{7/4}} \sum_{\ell < N^\delta} 
 \Big\langle  \frac{a^* (\ph)^\ell}{\sqrt{\ell !}} \Omega, \cU_2^* (t;s) \, (\cK+\cN)  \, \cU_2 (t;s) \frac{a^* (\ph)^\ell}{\sqrt{\ell !}} \Omega \Big\rangle \\ &+ \frac{1}{N^{2+p/12} \alpha_N^2} \sum_{\ell < N^\delta}  \Big\langle  \frac{a^* (\ph)^\ell}{\sqrt{\ell !}} \Omega , \cU_2^* (t;s) \, \cN^{4+p}  \, \cU_2 (t;s) \frac{a^* (\ph)^\ell}{\sqrt{\ell !}} \Omega \Big\rangle \, , \end{split} \]
for arbitrary $p \geq 0$ (we use here the fact that ${\bf 1} (\cN \geq N^{1/12}) \leq N^{-p/12} \, \cN^p$, for any $p \geq 0$). Combining Lemma \ref{lm:7}, Lemma \ref{lm:L2} and Lemma \ref{lm:1}, we conclude that
\[  \begin{split}
\sum_{\ell < N^\delta}  \Big\langle  &\frac{a^* (\ph)^\ell}{\sqrt{\ell !}} \Omega,  \, \cU_2^* (t;s) \, \cL_4^2 \, \cU_2 (t;s) \frac{a^* (\ph)^\ell}{\sqrt{\ell !}} \Omega \Big\rangle  \\  \leq \; & \frac{C e^{\wt{K} (t-s)}}{N^{7/4}} \sum_{\ell < N^\delta}  \Big\langle  \frac{a^* (\ph)^\ell}{\sqrt{\ell !}} \Omega , (\cK+\cN)  \, \frac{a^* (\ph)^\ell}{\sqrt{\ell !}} \Omega \Big\rangle 
+ \frac{C e^{\wt{K} (t-s)}}{N^{2+p/12} \alpha_N^2} \sum_{\ell < N^\delta}  \Big\langle  \frac{a^* (\ph)^\ell}{\sqrt{\ell !}} \Omega,  \, \cN^{4+p}  \,\frac{a^* (\ph)^\ell}{\sqrt{\ell !}} \Omega \Big\rangle \\
 \leq \; & \frac{C e^{\wt{K} (t-s)}}{N^{7/4}} \| \wt{\ph}_t \|_{H^1} \sum_{\ell < N^\delta}  \ell + \frac{C e^{\wt{K} (t-s)}}{N^{2+p/12-(5+p)\delta} \alpha_N^2} \, \sum_{\ell < N^\delta} \ell^{4+p} \\
 \leq \; & \frac{C e^{\wt{K} (t-s)}}{N^{7/4-2\delta}}  + \frac{C e^{\wt{K} (t-s)}}{N^{2+p/12-(5+p)\delta} \alpha_N^2} \, ,
\end{split} \] 
where we used that $\| \wt \ph_t \|_{H^1}$ remains uniformly bounded in $t \in \bR$ (by a constant depending only on $\| \ph \|_{H^1}$). Together with (\ref{eq:B21}), we obtain, from (\ref{eq:B2}), 
\[ \begin{split}  \Big\| (\cN+1)^{-\kappa} \, \cL_4 \, & \cU_2 (t;s) W^* (\sqrt{N} \ph) 
\frac{a^* (\ph)^N}{\sqrt{N!}} \Omega \Big\|^2 \\ \leq \; &C e^{\wt{K} (t-s)} \left(\frac{1}{N^{2+\delta (2\kappa-6)} \alpha_N^2} +  \frac{1}{N^{7/4-2\delta}}  + \frac{1}{N^{2+p/12-(5+p)\delta} \alpha_N^2} \right)\,.  \end{split} \]
We fix $\delta < 1/12$. Moreover, we choose $\kappa >0$ so large that $N^{2+ \delta (2\kappa - 6)} \alpha_N^2 \geq N^{3/2}$, and $p >0$ so large that $N^{1+p (1/12 - \delta)} \alpha_N^2 \geq N^{3/2}$. Then, together with (\ref{eq:B0}), we find that the term $\text{B}$ on the r.h.s. of (\ref{eq:AB}) can be bounded by 
\[ \text{B} \leq  \frac{C \| J \| e^{K t}}{d_N \sqrt{N}} \]
because $d_N \leq C N^{1/4}$ (recall that $d_N = \sqrt{N!} / (e^{-N/2} N^{N/2})$). Together with (\ref{eq:A-pr}), this completes the proof of the proposition. 
\end{proof}

The next lemma is used to bound the kinetic energy by the generator $\cL_2 (t)$ of the dynamics $\cU_2 (t;s)$. 
\begin{lem} \label{lm:7}
Let $\cL_2 (t)$ be as defined in (\ref{eq:L2}) and let $\cK$ be the kinetic energy operator 
\begin{equation}
\cK = \int dx \; \nabla_x a_x^* \, \nabla_x a_x.
\end{equation}
Then there exists a constant $C>0$ such that the operator inequalities
\begin{equation}
-C (\N+1) \leq \cL_2(t) - \cK \leq C (\N+1)
\end{equation}
hold true for all $t \in \bR$.
\end{lem}

\begin{proof}
It suffices to show that there exists a constant $C$ such that
\begin{equation} \label{H2-H_kin estimate}
\left| \langle \psi, (\cL_2 (t) - \cK) \psi \rangle \right| \leq C \langle \psi, (\N+1) \psi \rangle
\end{equation}
for all $\psi \in \F$. By definition, we have that
\begin{equation} \label{H2-H_kin} \begin{split}
\langle \psi, (\cL_2 (t) - \cK) \psi \rangle &= \int dx (\widetilde{V}*|\wt \ph_t|^2)(x) \langle a_x \psi, a_x \psi \rangle + \int dx dy \; \widetilde{V}(x-y) \overline{\wt \ph}_t (x) \wt \ph_t (y) \langle a_y \psi, a_x \psi \rangle \\
& \qquad + 2 \text{Re } \int dx dy \; \widetilde{V}(x-y) \wt \ph_t (x) \wt \ph_t (y) \langle a_y \psi, a_x^{*} \psi \rangle \, . 
\end{split} \end{equation}
The first term on the r.h.s. of the last equation can be estimated by
\begin{equation}
\left| \int dx (\widetilde{V}*|\wt \ph_t|^2)(x) \langle a_x \psi, a_x \psi \rangle \right| \leq \| \widetilde{V}* |\wt \ph_t|^2 \|_{\infty} \int dx \| a_x \psi \|^2 \leq C \langle \psi, \N \psi \rangle.
\end{equation}
The second term on the r.h.s. of (\ref{H2-H_kin}) can be handled similarly:
\begin{equation} \begin{split}
\left| \int dx dy \; \widetilde{V}(x-y) \overline{\wt \ph}_t (x) \wt \ph_t (y) \langle a_y \psi, a_x \psi \rangle \right| & \leq \int dx dy \; |\widetilde{V}(x-y)| |\wt \ph(y)|^2 \| a_x \psi \|^2 \\
&\leq \| |\widetilde{V}| * |\wt \ph_t|^2 \|_{\infty} \int dx \| a_x \psi \|^2 \leq C \langle \psi, \N \psi \rangle.
\end{split} \end{equation}
Finally, the last term on the r.h.s. of (\ref{H2-H_kin}) is bounded by 
\begin{equation} \begin{split}
 \left| \int dx dy \; \widetilde{V}(x-y) \wt \ph_t (x) \wt \ph_t (y) \langle a_y \psi, a_x^{*} \psi \rangle \right| &= \left| \int dy \,  \wt \ph_t (y) \, \langle a_y \psi, a^* (\wt{V} (y-.)\wt \ph_t) \psi \rangle \right| \\
&\leq \int dy \| a_y \psi \|^2 + \int dy \; |\wt \ph_t (y)|^2  \| a^* (\widetilde{V}(y-.) \wt \ph_t)\psi \|^2 \\ &\leq C \left(1 + \sup_y \, \| \wt{V} (y-.) \wt \ph_t \| \right) \, \langle \psi, (\N+1) \psi \rangle \\ &\leq C \langle \psi, (\N+1) \psi \rangle , 
\end{split} \end{equation}
where we used Lemma \ref{lm:0}.
\end{proof}

After controlling the kinetic energy with the expectation of the generator $\cL_2 (t)$, we have to show that this expectation remains bounded in time. This is the content of the next lemma.
\begin{lem} \label{lm:L2}
Let $\cU_2 (t;s)$ be the evolution defined in (\ref{eq:U2}), with generator $\cL_2 (t)$. 
There exist constants $C$ and $K$ such that
\begin{equation}
\left| \langle \U_2 (t;s) \psi, \cL_2 (t) \U_2 (t;s) \psi \rangle \right| \leq C e^{K|t-s|} \langle \psi, (\cL_2 (s) + \cN +1 ) \psi \rangle
\end{equation}
for all $\psi \in \F$ and all $t,s \in \bR$.
\end{lem}

\begin{proof}
We use the shorthand notation $\psi_2 = \U_2 (t;s) \psi$. To control $\langle \psi_2, \cL_2 (t) \psi_2 \rangle$, we first observe that
\begin{equation}
\frac{d}{dt} \langle \psi_2, \cL_2 (t) \psi_2 \rangle = \langle \psi_2, \dot{\cL}_2 (t) \psi_2 \rangle
\end{equation}
with the time-derivative 
\begin{equation} \label{eq:L2dot} \begin{split}
\dot{\cL}_2 (t) = \; &\int dx dy \widetilde{V}(x-y) \left( \overline{\wt \ph}_t (y) \dot{\wt \ph}_t (y)  + \dot{\overline{\wt \ph}}_t (y) \wt \ph_t (y) \right) a_x^{*} a_x \\ &+ \int dx dy \widetilde{V}(x-y) \left( \overline{\wt \ph}_t (x) \dot{\wt \ph}_t (y) + \dot{\overline{\wt \ph}}_t (x) \wt \ph_t (y) \right) a_y^{*} a_x  \\
&+ 2 \int dx dy \widetilde{V}(x-y) \left( \wt \ph_t (x) \dot{\wt \ph}_t (y)  a_x^{*} a_y^{*} + \overline{\wt \ph}_t (x) \dot{\overline{\wt \ph}}_t (y) a_x a_y \right) \, .
\end{split}
\end{equation}
Next, we want to control $\langle \psi_2, \dot{\cL}_2 (t) \psi_2 \rangle$ in terms of $\langle \psi_2 , (\cL_2 (t) +\cN) \psi_2 \rangle$. There are several contributions to the expectation $\langle \psi_2 , \dot{\cL}_2 (t) \psi_2 \rangle$ arising from the terms on the r.h.s. of (\ref{eq:L2dot}). For example, the contribution from the last line of (\ref{eq:L2dot}) is given by
\begin{equation}\label{eq:I-H2}
\begin{split}
\text{I} = \; &4 \text{Re} \int dx dy\,  \widetilde{V}(x-y) \wt \ph_t (x) \dot{\wt \ph_t}(y) \langle \psi_2, a_x^{*} a_y^{*} \psi_2 \rangle \\
=\; &  - 4 \text{Im } \int dx dy \, \widetilde{V}(x-y) \wt \ph_t (x) \Delta \wt \ph_t (y) \, \langle a_x a_y \psi_2, \psi_2 \rangle \\
&+ 4 \text{Im } \int dx dy \, \widetilde{V}(x-y) \wt \ph_t (x) (\widetilde{V} * |\wt \ph_t|^2 )(y) \wt \ph_t (y)
\,  \langle a_x a_y \psi_2, \psi_2 \rangle 
\end{split} \end{equation}
since $\wt \ph_t$ solves the regularized Hartree equation (\ref{eq:hartreereg}).
Since $\| \widetilde{V} * |\wt \ph_t|^2 \|_{\infty} \leq \| \wt \ph_t \|_{H^1}^2 \leq C$, the second line on the r.h.s. of the last equation can be bounded by 
\begin{equation} \label{eq:Ia-H2} \begin{split}
\Big| \int dy \, (\widetilde{V} * |\wt \ph_t|^2) (y) \, & \wt \ph_t (y) \, \langle a_y \psi_2, a^* (\wt{V} (y-.)\wt \ph_t) \psi_2 \rangle \Big| \\ 
\leq \; &  \int dy \,  \, \left|(\widetilde{V} * |\wt \ph_t|^2) (y)\right|^2 \, | \wt \ph_t (y)|^2 \, \| a^*  (\wt{V} (.-y)\wt \ph_t) \psi_2 \|^2  +  \int dy  \, \| a_y \psi_2 \|^2 \\ 
\leq \; &  \left( 1 + \| \widetilde{V} * |\wt \ph_t|^2 \|_{\infty}^2 \,\sup_y \| \wt{V} (y-.) \wt \ph_t \|^2_2 \right)  \| (\cN+1)^{1/2} \psi_2 \|^2 \\
\leq \; & C \langle \psi_2, (\cN+1) \psi_2 \rangle \, . 
\end{split} \end{equation}

As for the first term on the r.h.s. of (\ref{eq:I-H2}), we write $\psi_2 = \{ \psi_2^{(n)} \}_{n\geq 0}$ and
\[ \begin{split} \langle a_x a_y \psi_2 , \psi_2 \rangle = & \; \sum_{n \geq 0} \int dx_1 \dots dx_n \, \overline{( a_x a_y \psi_2)^{(n)}} (x_1, \dots , x_n) \, \psi_2^{(n)} (x_1, \dots , x_n) 
\\ = & \; \sum_{n \geq 0} \sqrt{(n+1)(n+2)}  \int dx_1 \dots dx_n \, \overline{\psi}_2^{(n+2)} (x,y, x_1, \dots , x_n) \, \psi_2^{(n)} (x_1, \dots , x_n) \,.
\end{split} \]
Therefore, introducing the notation $\bx_n = (x_1, \dots, x_n)$, 
\begin{equation} \begin{split}
\int dx dy \, &\widetilde{V}(x-y) \wt \ph_t (x) \Delta \wt \ph_t (y) \, \langle a_x a_y \psi_2, \psi_2 \rangle \\ = \; & \sum_{n \geq 0} \sqrt{(n+1)(n+2)} \int dx dy d\bx_n \,  \widetilde{V}(x-y) \wt \ph_t (x) \Delta \wt \ph_t (y) \, \overline{\psi}_2^{(n+2)} (x,y,\bx_n) \, \psi_2^{(n)} (\bx_n) . \end{split} 
\end{equation}
Integrating by parts, we find
\[ \begin{split}
\int dx dy \, &\widetilde{V}(x-y) \wt \ph_t (x) \Delta \wt \ph_t (y) \, \langle a_x a_y \psi_2, \psi_2 \rangle \\ = \; & \sum_{n \geq 0} \sqrt{(n+1)(n+2)} \int dx dy d\bx_n \, \nabla \widetilde{V}(x-y) \wt \ph_t (x) \nabla\wt \ph_t (y) \, \overline{\psi}_2^{(n+2)} (x,y, \bx_n) \, \psi_2^{(n)} (\bx_n) \\ &- \sum_{n \geq 0} \sqrt{(n+1)(n+2)} \int d\bx_n \, \widetilde{V}(x-y) \wt \ph_t (x) \nabla\wt \ph_t (y) \, \nabla_y \overline{\psi}_2^{(n+2)} (x,y, \bx_n) \, \psi_2^{(n)} (\bx_n) .
\end{split}
\]
In the first term, we integrate by parts once more, but this time w.r.t. the variable $x$:
\[ \begin{split}
\int dx dy \, &\widetilde{V}(x-y) \wt \ph_t (x) \Delta \wt \ph_t (y) \, \langle a_x a_y \psi_2, \psi_2 \rangle \\ = \; &- \sum_{n \geq 0} \sqrt{(n+1)(n+2)} \int dx dy d\bx_n \, \widetilde{V}(x-y) \nabla\wt \ph_t (x) \nabla\wt \ph_t (y) \, \overline{\psi}_2^{(n+2)} (x,y, \bx_n) \, \psi_2^{(n)} (\bx_n) \\ &- \sum_{n \geq 0} \sqrt{(n+1)(n+2)} \int dx dy d\bx_n \, \widetilde{V}(x-y) \wt \ph_t (x) \nabla\wt \ph_t (y) \, \nabla_x \overline{\psi}_2^{(n+2)} (x,y, \bx_n) \, \psi_2^{(n)} (\bx_n)
\\ &- \sum_{n \geq 0} \sqrt{(n+1)(n+2)} \int dx dy d\bx_n \, \widetilde{V}(x-y) \wt \ph_t (x) \nabla\wt \ph_t (y) \, \nabla_y  \overline{\psi}_2^{(n+2)} (x,y,\bx_n) \, \psi_2^{(n)} (\bx_n).
\end{split}
\]
Taking absolute value, and using Cauchy-Schwarz, we find
\[ \begin{split}
\Big|\int dx dy \, \widetilde{V}(x-y) &\wt \ph_t (x) \Delta \wt \ph_t (y) \, \langle a_x a_y \psi_2, \psi_2 \rangle \Big| \\ \leq \; & \sum_{n \geq 0} (n+2) \int dx dy d\bx_n \, \widetilde{V}^2 (x-y) |\psi_2^{(n+2)} (x,y, \bx_n)|^2 \\ &+\sum_{n \geq 0} (n+1) \int dx dy d\bx_n \, |\nabla\wt \ph_t (x)|^2 \, |\nabla\wt \ph_t (y)|^2 \, 
|\psi_2^{(n)} (\bx_n)|^2 \\
&+ 2 \sum_{n \geq 0} (n+2) \int dx dy d\bx_n \, |\nabla_x \psi_2^{(n+2)} (x,y, \bx_n)|^2 \\
&+ 2 \sum_{n \geq 0} (n+1) \int dx dy d\bx_n \, \wt{V}^2 (x-y) |\wt \ph_t (x)|^2 \, |\nabla\wt \ph_t (y)|^2 \,  |\psi_2^{(n)} (\bx_n)|^2 .
\end{split}\]
Using the fact that $\| \wt{V}^2 *|\wt \ph_t|^2 \|_{\infty} \leq \| \wt \ph_t \|_{H^1}^2 \leq C$ for all $t \in \bR$, and since 
\[  \int dx dy d\bx_n \, \widetilde{V}^2 (x-y) |\psi_2^{(n+2)} (x,y, \bx_n)|^2 \leq \int dx dy d\bx_n \, \left(|\nabla_x \psi_2^{(n+2)} (x,y, \bx_n)|^2 + |\psi_2^{(n+2)} (x,y,\bx_n)|^2 \right) \]
we conclude that
\[ \Big|\int dx dy \, \widetilde{V}(x-y) \wt \ph_t (x) \Delta \wt \ph_t (y) \, \langle a_x a_y \psi_2, \psi_2 \rangle \Big| \leq C \langle \psi_2 , (\cK + \cN + 1) \psi_2 \rangle  \, .\]
Together with (\ref{eq:Ia-H2}) and (\ref{eq:I-H2}), this implies that
\[ | \text{I} | \leq C  \langle \psi_2 , (\cK + \cN + 1) \psi_2 \rangle \, . \] 

The contribution from the second line on the r.h.s. of (\ref{eq:L2dot}) can be bounded analogously. The contribution from the first term on the r.h.s. of (\ref{eq:L2dot}) is given by
\[ \begin{split} 
2 \text{Re } \int dx dy \, \wt{V} (x-y) \overline{\wt \ph}_t (y) \dot{\wt \ph}_t (y) \| a_x \psi_2 \|^2   = \; &2 \text{Im } \int dx dy \wt{V} (x-y) \overline{\wt \ph}_t (y) \Delta \wt \ph_t (y) \, \| a_x \psi_2 \|^2 \\ &+ 2 \text{Im } \int dx dy \wt{V} (x-y) |\wt \ph_t (y)|^2 (\wt{V} * |\wt \ph_t|^2) (y) \, \| a_x \psi_2 \|^2 \\ = \; &\text{A } + \text{B}  \, .
\end{split}
\]
The second term can be estimated by
\[ |\text{B}| \leq 2 \| \wt{V} * |\wt \ph_t|^2 \|^2_{\infty} \langle \psi_2 , \cN \psi_2 \rangle \, . \]
As for the first term, we integrate by parts. Since
\[ \int dx dy \wt{V} (x-y) |\nabla \wt \ph_t (y)|^2 \, \| a_x \psi_2 \|^2 \] is clearly a real number, we find
\[ \text{A} = 2 \text{Im } \int dx dy \nabla\wt{V} (x-y) \overline{\wt \ph}_t (y) \nabla \wt \ph_t (y) \, \| a_x \psi_2 \|^2 \, . \]
Integrating by parts with respect to $x$, we conclude that
\[ \begin{split} |\text{A}| \leq \; &4  \int dx dy \, |\wt{V} (x-y)| \, |\wt \ph_t (y)| |\nabla \wt \ph_t (y)| \, \| \nabla_x a_x \psi_2 \| \, \| a_x \psi_2 \| \\ \leq \; &4 \int dx dy \,|\wt{V} (x-y)|^2 \, |\wt \ph_t (y)|^2 \| \nabla_x a_x \psi_2 \|^2 + \int dx dy \, |\nabla \wt \ph_t (y)| \,  \| a_x \psi_2 \|^2
\\ \leq \; &4 \| \wt{V} * |\wt \ph_t|^2 \|_\infty \langle \psi_2, \cK \psi_2 \rangle + \| \wt \ph_t \|_{H^1} \langle \psi_2 , \cN \psi_2 \rangle \\ \leq \; & C \langle \psi_2, (\cK + \cN) \psi_2 \rangle.
\end{split} \]

Summarizing, we showed that
\[ \left| \frac{d}{dt}  \langle \psi_2 , \cL_2 (t) \psi_2 \rangle \right| \leq C \langle \psi_2, (\cK+ \cN + 1) \psi_2 \rangle \, . \]
Together with Lemma \ref{lm:7} and with Proposition \ref{lm:1}, we conclude that
\[ \left| \frac{d}{dt}  \langle \psi_2 , (\cL_2 (t) + \cN + 1) \psi_2 \rangle \right| \leq C \langle \psi_2, (\cL_2 (t) + \cN + 1) \psi_2 \rangle  \, .\]
Hence, the lemma follows from Gronwall inequality.
\end{proof}

Lemma \ref{lm:7} and Lemma \ref{lm:L2} allow us to control the contribution of the term with $\cL_4$ in (\ref{eq:L3+L4}). To bound the contribution containing $\cL_3 (\tau)$, we use the following lemma. 
\begin{lem}\label{lm:L3}
Suppose $\cL_3 (t)$ is defined as in (\ref{eq:defL3L4}). Then there exists a constant $C>0$ such that 
\[ \| (\cN+1)^j \cL_3 (t) \psi \| \leq \frac{C}{\sqrt{N}} \| (\cN+1)^{j+3/2} \psi \| \]
for all $t \in \bR$. 
\end{lem}

\begin{proof}
We compute
\begin{equation} \label{eq:I+II}\begin{split}
\| \cN^j \cL_3 (t) \psi \|^2 = \; & \langle \psi, \cL_3 (t) \cN^{2j} \cL_3 (t) \psi \rangle \\ = \; &\int dx dy dz dw \, \wt V(x-y) \, \wt V(z-w) \, \\ &\hspace{1.5cm} \times \langle \psi , a_x^* (a_y^* \wt \ph_t (y) + a_y \overline{\wt \ph}_t (y) ) a_x \cN^{2j} a_z^* \left( a_w^* \wt \ph_t (w) + a_w \overline{\wt \ph}_t (w) \right) a_z \psi \rangle \\ = \; & 2 \text{Re } \int dx dy dz dw \, \wt V(x-y) \, \wt V(z-w) \,\wt \ph_t (y) \wt \ph_t (w)  \langle \psi , a_x^* a_y^* a_x \cN^{2j} a_z^* a_w^* a_z \psi \rangle \\ &+ 2 \text{Re } \int dx dy dz dw \, \wt V(x-y) \, \wt V(z-w) \,\wt \ph_t (y) \overline{\wt \ph}_t (w)  \langle \psi , a_x^* a_y^* a_x \cN^{2j} a_z^* a_w a_z \psi \rangle \\ =\; & \text{I } + \text{II} \, .
\end{split} \end{equation}
Using the canonical commutation relations and the formula $a_x \cN = (\cN+1) a_x$, we find 
\begin{equation}\label{eq:I} \begin{split}
\text{I } =\; & 2 \text{Re }  \int dx dy dz dw \, \wt V(x-y) \, \wt V(z-w) \,\wt \ph_t (y) \wt \ph_t (w)  \langle 
a_x a_y a_z a_w (\cN+1)^{j-1/2} \psi , a_x a_z (\cN + 1)^{j+1/2} \psi \rangle \\ &+ 2 \text{Re } \int dx dy dw \,  \wt V(x-y) \, \wt V(x-w) \,\wt \ph_t (y) \wt \ph_t (w)  \langle 
a_x a_y a_w (\cN+1)^{j-1/2} \psi , a_x (\cN + 1)^{j+1/2} \psi \rangle \\ &+ 2 
\text{Re } \int dx dy dz \,  \wt V(x-y) \, \wt V(z-x) \,\wt \ph_t (y) \wt \ph_t (x)  \langle 
a_x a_y a_z (\cN+1)^{j-1/2} \psi , a_z (\cN + 1)^{j+1/2} \psi \rangle \\ = \; & \text{A } + \text{B } + \text{C}  \, .\end{split} \end{equation}
Applying Schwarz inequality, we find 
\[ \begin{split} 
|\text{A}| \leq \; &  \int dx dy dz dw \, \| a_x a_y a_z a_w (\cN+1)^{j-1/2} \psi \|^2 \\ &+ \int dx dy dz dw \, \wt V^2 (x-y) \, \wt V^2 (z-w) \, |\wt \ph_t (y)|^2  |\wt \ph_t (w)|^2 \| a_x a_z (\cN + 1)^{j+1/2} \psi \|^2 \\ \leq \; & \| (\cN + 1)^{j+3/2} \psi \|^2 \left( 1 + \| \wt V^2 * |\wt \ph_t (y)|^2 \|_\infty^2 \right) \\ \leq \; &C  \| (\cN + 1)^{j+3/2} \psi \|^2
\end{split} \]
where we used that, since $\wt V^2 \leq (1-\Delta)$, $\| \wt V^2 * |\wt \ph_t|^2 \|_\infty  \leq \| \wt \ph_t \|_{H^1} \leq C$ uniformly in $t \in \bR$. The second term on the r.h.s. of (\ref{eq:I}) can be bounded by 
\[ \begin{split} 
|\text{B}| \leq  \; & \int dx dy dw \,  \|  
a_x a_y a_w (\cN+1)^{j-1/2} \psi \|^2 \\ &+ \int dx dy dw \, \wt V^2 (x-y) \, \wt V^2 (x-w) \, |\wt \ph_t (y)|^2 |\wt \ph_t (w)|^2   \| a_x (\cN + 1)^{j+1/2} \psi \|^2   \\  \leq \; &C  \| (\cN + 1)^{j+1} \psi \|^2  \, . \end{split} \]
Similarly, the third term on the r.h.s. of (\ref{eq:I}) is controlled by 
\[\begin{split} 
| \text{C}| \leq \; & \int dx dy dz \,  \|  
a_x a_y a_z (\cN+1)^{j-1/2} \psi \|^2 \\ &+ \int dx dy dz \, \wt V^2 (x-y) \, \wt V^2 (z-x) \, |\wt \ph_t (y)|^2 |\wt \ph_t (x)|^2   \| a_z (\cN + 1)^{j+1/2} \psi \|^2   \\ \leq \; &C  \| (\cN + 1)^{j+1} \psi \|^2 \end{split} \]
where, on the second line, we first integrate over $y$ and we extract the supremum over $x$ of $|\wt V^2 * |\wt \ph_t|^2 )(x)|$. Afterwards we integrate over $x$ and extract the supremum over $z$ of $|(\wt V^2 * |\wt \ph_t|^2)(z)|$ and, finally, we integrate over $z$.

The second term on the r.h.s. of (\ref{eq:I+II}) can be written as 
\[ \begin{split}
\text{II } = \; & 2 \text{Re }  \int  dx dy dz dw \, \wt V(x-y) \, \wt V(z-w) \,\wt \ph_t (y) \overline{\wt \ph}_t (w)  \langle a_x a_y a_z  (\cN-1)^j \psi , a_x a_w a_z  (\cN-1)^j \psi \rangle \\ &+ 2 \text{Re }  \int  dx dy dw \, \wt V(x-y) \, \wt V(x-w) \,\wt \ph_t (y) \overline{\wt \ph}_t (w)  \langle a_x a_y (\cN-1)^j \psi , a_x a_w  (\cN-1)^j \psi \rangle  \, .
\end{split} 
\]
Hence, we can estimate 
\[ \begin{split} 
| \text{II}| \leq \; & 2 \int dx dy dz dw \, \wt V^2 (z-w) |\wt \ph_t (w)|^2 \, \| a_x a_y a_z (\cN-1)^j \psi \|^2\\ &+  2 \int  dx dy dw \, \wt V^2 (x-w) \, |\wt \ph_t (w)|^2  \, \| a_x a_y (\cN-1)^j \psi \|^2 \\ \leq \; & \| \wt V^2 * |\wt \ph_t|^2 \|_{\infty} \, \| (\cN+1)^{j+3/2} \psi \|^2  \, .
\end{split} \]
This completes the proof of the lemma.
\end{proof}

\section{Relation between product states and coherent states}
\label{sec:pro-coh}

In this paper we are interested in the evolution of factorized initial data of the form $\ph^{\otimes N}$ with a fixed number of particles $N$. Since it is more convenient to work with coherent states, we write
\[ \{ 0, \dots , 0, \ph^{\otimes N}, 0, \dots \} = \frac{a^* (\ph)^{\otimes N}}{ \sqrt{N!}} = d_N P_N W(\sqrt{N} \ph) \Omega \]
where the constant $d_N \simeq N^{1/4}$ takes into account the fact that only a small part of the coherent state $W(\sqrt{N} \ph) \Omega$ lies in the $N$-particles sector. Similarly, if we apply the inverse Weyl operator $W^* (\sqrt{N} \ph)$ to the factorized state $\{ 0, \dots , 0, \ph^{\otimes N} , 0,\dots \}$, only a small part (of size $d_N^{-1}$) 
of the resulting Fock space vector will have a small number of particles. This is the content of the next lemma, whose proof can be found in \cite{CL}. 
\begin{lem} \label{lm:2}
There exists a constant $C>0$ such that, for any $\ph \in L^2 (\bR^3, dx)$, we have
\[ \left\| (\cN+1)^{-1/2} W^* (\sqrt{N} \ph) \frac{a^* (\ph)^N}{\sqrt{N!}} \Omega \right\| \leq \frac{C}{d_N}  \, . \]
\end{lem}

\section{A property of the quadratic evolution $\cU_2 (t;s)$}
\label{sec:U2}

One of the reasons why we obtain precise error bounds is the observation that, for arbitrary $t,s \in \bR$ and $f \in L^2 (\bR^3, dx)$, the vectors $\cU_2 (t;s) a^* (f) \cU_2 (t;s) \Omega$ and $\cU_2 (t;s) a (f) \cU_2 (t;s) \Omega$ are localized in the one-particle sector. This fact is proven in the following lemma. 

\begin{lem}\label{lm:3} Suppose that the evolution $\cU_2 (t;s)$ is defined as in (\ref{eq:U2}). Then we have, for any $f \in L^2 (\bR^3, dx)$ and any $t \in \bR$, 
\[ \cU_2 (t;0)^* \phi (f) \cU_2 (t;0) \Omega = P_1 U_2 (t;0)^* \phi (f) \cU_2 (t;0) \Omega  \, . \]
\end{lem} 

\begin{proof}
For any $\psi \in  \cF$ with $\| \psi \| = 1$ and $\psi = {\bf 1} (\cN = m) \psi$, with $m \not = 1$, we define the quantity
\[ F(t) = \sup_{f \in L^2 (\bR^3)} \frac{1}{\| f \|} \left| \langle \psi , \cU_2 (t;0)^* a(f) \cU_2 (t;0) \Omega \rangle \right| + \sup_{f \in L^2 (\bR^3)}  \frac{1}{\| f \|} \left| \langle \psi , \cU_2 (t;0)^* a^* (f) \cU_2 (t;0) \Omega \rangle \right| \, .
\] Note that $F(0) = 0$. Let 
\[ \cK = \int dx \, \nabla_x a^*_x \nabla_x a_x  \, .\]
We observe that, for any $f \in L^2 (\bR^3, dx)$, 
\[ e^{i\cK t} a(f) e^{-i\cK t} = a (e^{-i\Delta t} f)  \, .  \]
Since $e^{-i\Delta t}$ is a unitary operator on $L^2 (\bR^3, dx)$, we conclude that
\[ \sup_{f \in L^2 (\bR^3)} \frac{1}{\| f \|} \left| \langle \psi , \cU_2 (t;0)^* a(f) \cU_2 (t;0) \Omega \rangle \right| = \sup_{f \in L^2 (\bR^3)} \frac{1}{\| f \|} \left|\langle \psi , \cU_2 (t;0)^* e^{i\cK t} a(f) e^{-i\cK t} \cU_2 (t;0) \Omega \rangle \right|  \] and similarly if we replace $a(f)$ with $a^* (f)$. This implies that
\[\begin{split}  F(t) = \; & \sup_{f \in L^2 (\bR^3)}  \frac{1}{\| f \|} \left| \langle \psi , \cU_2 (t;0)^* e^{i\cK t} a(f) e^{-i\cK t} \cU_2 (t;0) \Omega \rangle \right| \\ &+ \sup_{f \in L^2 (\bR^3)}  \frac{1}{\| f \|} \left| \langle \psi , \cU_2 (t;0)^* e^{i\cK t} a^* (f) e^{-i\cK t} \cU_2 (t;0) \Omega \rangle \right| \, .  \end{split}
\]

For $f \in L^2(\bR^3, dx)$, we compute
\[ \begin{split} i\frac{d}{dt}  \langle \psi , \cU_2 (t;0)^* &e^{i\cK t} a(f) e^{-i\cK t} \cU_2 (t;0) \Omega \rangle \\ & = \langle \psi, \cU_2 (t;0)^* \left[  e^{i\cK t} a(f) e^{-i\cK t},  \cL_2 (t) - \cK \right]  \cU_2 (t;0) \Omega \rangle \\ &
 = \langle \psi, \cU_2 (t;0)^* \left[  a(f_t), \cL_2 (t) - \cK \right]  \cU_2 (t;0) \Omega \rangle \end{split} \]
 with $f_{t} = e^{-i\Delta t} f$.  Using the canonical commutation relations, it is simple to check that
\[ [ a(f_t) , \cL_2 (t) - \cK ] = a ( (\wt V*|\wt \ph_t|^2) f_t + (\wt V*f_t \overline{\wt \ph_t}) \wt \ph_t) + a^* (2 (\wt V*\overline{f}_t \wt \ph_t ) \wt \ph_t ) \, . \]

Notice that, under the assumption $\wt V^2 \leq C (1-\Delta)$, we find
\[ \| (\wt V*|\wt \ph_t|^2) f_t \| \leq \| f_t \| \sup_x \int dy \wt V(x-y) |\wt \ph_t (y)|^2 \leq C \| f_t \| \|\wt \ph_t \|_{H^{1/2}}^2 \leq C \| f \| \]
and 
\[  \| (\wt V*f_t \overline{\wt \ph_t}) \wt \ph_t \| \leq \| \wt \ph_t \| \sup_x \int dy \, \wt V(x-y) |f_t(y)| |\wt \ph_t (y)|
\leq C \| f_t \| \| \wt \ph_t \|_{H^1}^2 \leq C \| f \| \]
for a constant $C$, independent of $t$ and $f$. 
Therefore, we conclude that 
\[  
\frac{1}{\| f \|} \left| \langle \psi , \cU_2 (t;0)^* a(f) \cU_2 (t;0) \Omega \rangle \right| \leq C \int_0^t ds \, F(s) \]
for every $f \in L^2 (\bR^3, dx)$. The same bound can be obtained with $a(f)$ replaced by $a^* (f)$. Hence, we obtain 
\[ 0 \leq F(t) \leq C \int_0^t ds F(s) \, . \]
This bound, together with $F(0) = 0$ and with the a-priori bound (which follows from Proposition \ref{lm:1}) 
\[ F(t) \leq 2 \| (\cN+1)^{1/2} \cU_2 (t;0) \Omega \|  \leq C e^{K |t|} \]
implies that $F(t) = 0$ for all $t \in \bR$. 
\end{proof}

\section*{Acknowledgment}
We are grateful to H.-T. Yau for helpful discussions.

\end{document}